\documentclass[onecolumn]{IEEEtran}

\usepackage{cite}
\usepackage{amsmath}
\usepackage{amsfonts}
\usepackage{pifont}
\usepackage{amssymb}
\usepackage{amsthm}
\usepackage{tikz}
\usepackage{cases}
\usepackage{graphicx}
    \graphicspath{{../}}
    \DeclareGraphicsExtensions{.pdf}
\usepackage[caption=false,font=footnotesize]{subfig}
\usepackage{multirow}
\usepackage{booktabs}
\usepackage{url}
\usepackage{xtab}
\usepackage{tabu}
\usepackage{longtable}
\usepackage{algorithm}
\usepackage{algorithmic}
\usepackage{enumerate}
\usepackage{makecell}
\usepackage{lipsum}
\usepackage{multicol}
\usepackage{mathdots}
\usepackage{extarrows}
\usepackage{color,xcolor}
\usetikzlibrary{arrows}
\theoremstyle{plain}

\newtheorem{theorem}{Theorem}
\newtheorem{lemma}[theorem]{Lemma}
\newtheorem{proposition}[theorem]{Proposition}
\newtheorem{definition}[theorem]{Definition}
\newtheorem{corollary}[theorem]{Corollary}
\theoremstyle{definition}
\newtheorem{example}{Example}

\newtheorem{remark}{Remark}

\begin{document}
\title{Capacity-Achieving PIR Schemes with Optimal Sub-Packetization}

\author{\IEEEauthorblockN{Zhifang Zhang,~Jingke Xu}\\
\IEEEauthorblockA{\fontsize{9.8}{12}\selectfont KLMM, Academy of Mathematics and Systems Science, Chinese Academy of Sciences, Beijing 100190, China\\
School of Mathematical Sciences, University of Chinese Academy of Sciences, Beijing 100049, China\\
Emails: zfz@amss.ac.cn, xujingke14@mails.ucas.edu.cn}

}
\maketitle
\thispagestyle{empty}

\begin{abstract}
Suppose a database containing $M$ records is replicated across $N$ servers, and a user wants to privately retrieve one record by accessing the servers such that identity of the retrieved record is secret against any up to $T$ servers. A scheme designed for this purpose is called a private information retrieval (PIR) scheme. In practice, capacity-achieving and small sub-packetization are both desired for PIR schemes, because the former implies the highest download rate and the latter usually means simple realization.

For general values of $N,T,M$, the only known capacity-achieving PIR scheme was designed by Sun and Jafar in 2016 with sub-packetization $N^M$. In this paper, we design a linear capacity-achieving PIR scheme with much smaller sub-packetization $dn^{M-1}$, where $d={\rm gcd}(N,T)$ and $n=N/d$. Furthermore, we prove that for any linear capacity-achieving PIR scheme it must have sub-packetization no less than $dn^{M-1}$, implying our scheme has the optimal sub-packetization. Moreover, comparing with Sun and Jafar's scheme, our scheme reduces the field size by a factor of $\frac{1}{Nd^{M-2}}$.

\end{abstract}

\section{Introduction}\label{se1}
According to the pioneer work by Chor et al. \cite{CKGS95FOCS:PIR} in 1995, the problem of private information retrieval (PIR) can be formulated as $N$ distributed servers each storing the replication of a database containing $M$ independent records $W_1,...,W_M$ and a user who wants to retrieve one record by accessing the servers while keeping identity of the retrieved record private from any individual server. Besides extensive applications in cryptographic protocols \cite{BFKR:PIR,GIKM:PIR}, PIR is also closely related to coding theory and theoretical computer science, e.g., locally decodable code \cite{Yekhanin07PHD:LDC&PIR} and one-way function \cite{Beimeletal99STOC:OWF&PIR}.

A central issue for the PIR problem is reducing communication cost. In the initial setting of PIR where each record is one bit long, communication cost is measured by the total number of bits transferred from the user to the servers (i.e. upload size) and bits from the servers to the user (i.e. download size). By now, the most efficient PIR schemes require communication cost $M^{O(\frac{1}{\log\log M})}$ \cite{Dvir&Gopi15STOC:2PIR,Efremenko09:LDC}. Determining the limits on communication cost for PIR remains an open problem.
However, it is more common in application scenarios that the size of each record is much larger than $M$ and $N$, thus the upload size is negligible compared to the download size. As a result, communication cost can be measured by taking only the  download size into account. Specifically, define the {\it rate} of a PIR scheme as the ratio between the  retrieved data size and the download size, and define  {\it capacity} as the supremum of the rate over all PIR schemes. In other words, capacity characterizes the maximum number of bits that can be privately retrieved per bit of downloaded data. Recently, much work has been done on determining the capacity of PIR in various cases.
Sun and Jafar first proved in \cite{Sun&Jafar16:CapacityPIR} the capacity of PIR for replicated servers is {\small$\frac{1-1/N}{1-(1/N)^M}$}. Then they
extended to considering PIR with privacy amplification secure against any colluding set containing at most $T$ servers, $1\leq T<N$. They proved in \cite{Sun&Jafar16:ColludPIR} the capacity for this extension is {\small$\frac{1-T/N}{1-(T/N)^{M}}$}. Meanwhile, the capacity for PIR with symmetric privacy \cite{Sun&Jafar16:CapaSymmPIR} and for PIR with MDS coded servers \cite{Bana&Uluk16:CapacityPIRCoded,Wang&Sko16:CapaSymmePIR} have been determined in some cases. For other variants of PIR, such as multi-message retrieval \cite{Bana&Uluk17:CapacityMPIR} and PIR from Byzantine servers \cite{Bana&Uluk17:CapacityBTPIR}, the problem of reducing download cost has also been considered.

Achievability of the capacity is usually proved by constructing  general capacity-achieving PIR schemes. These schemes are mostly implemented by dividing each record into segments and querying from each server some combinations of transformed segments. We call the number of segments contained in each record that are necessary for the scheme implementation as sub-packetization of the scheme. For linear schemes where only linear operations are involved, small sub-packetization usually implies simple scheme implementation. The problem of reducing sub-packetization has been deeply studied in the literature of minimum storage regenerating codes \cite{GTC:MSRsub,RTGE:MSRsub}. While for the PIR problem, study on sub-packetization of PIR schemes is scarce. On the one hand, some PIR schemes with small sub-packetization (e.g. linear sub-packetization) \cite{CHY15ISIT:Tradeoff,TR16:MDSPIR} are built at the sacrifice of non-achieving of capacity. On the other hand,  most known capacity-achieving schemes with asymmetric privacy have exponential sub-packetization. For example, the capacity-achieving schemes in \cite{Sun&Jafar16:CapacityPIR} and \cite{Sun&Jafar16:ColludPIR} have sub-packetization $N^M$. The lower bound on sub-packetization for PIR schemes is unknown in many cases.

\subsection{Related work}
In \cite{Sun&Jafar16:OptimalPIR} the authors characterized the optimal download cost for arbitrary record length and designed a PIR scheme attaining the capacity proved in \cite{Sun&Jafar16:CapacityPIR} with file size $N^{M-1}$.  Moreover, they proved this is the smallest possible file size for achieving capacity in this case. The file size  they discussed there actually means sub-packetization in our language. Thus this can be regarded as the first lower bound on sub-packetization for capacity-achieving PIR schemes with replicated servers and $T=1$ privacy.

\subsection{Our contribution}
In this work, we focus on sub-packetization for capacity-achieving PIR schemes with replicated servers and $T$ privacy, where $1\leq T<N$. We prove that sub-packetization for linear capacity-achieving PIR schemes is lower bounded by $dn^{M-1}$, where $d={\rm gcd}(N,T)$ and $n=N/d$. We further design a general linear capacity-achieving PIR schemes with
sub-packetization $dn^{M-1}$. Comparing with the scheme in \cite{Sun&Jafar16:ColludPIR}, we reduce the sub-packetization by a factor of $\frac{1}{nd^{M-1}}$, which leads to a great simplification in implementation. Based on the lower bound, our scheme actually attains the optimal sub-packetization. Moreover, a much smaller field (i.e. reducing from \cite{Sun&Jafar16:ColludPIR} by a factor of $\frac{1}{Nd^{M-2}}$)  is enough for our scheme.

For proving the lower bound on sub-packetization, we establish connections between the entropy of random variables and  the rank of corresponding linear objects. Based on these connections and some identities about information entropy for capacity-achieving PIR schemes, we derive several useful structural properties for linear PIR schemes. Our approach for proving the lower bound may be helpful for studying sub-packetization of PIR schemes in other cases.

The rest of the paper is organized as follows. In Section II we formally introduce the PIR model studied in this work and  reprove the capacity of PIR schemes. In Section III we prove a lower bound on sub-packetization for linear capacity-achieving PIR schemes. Then in Section IV we present a general linear capacity-achieving PIR scheme with optimal sub-packetization. Finally Section V concludes the paper.

\section{Preliminaries}
\subsection{Notations and the PIR model}\label{sec2}
For a positive integer $n\in\mathbb{N}$, we denote by $[n]$ the subset $\{1,...,n\}$. The vectors throughout this paper are always row vectors. For any vector $Q=(q_1,...,q_n)$ and any subset $\Gamma=\{i_1,...,i_m\}\subseteq [n]$, let $Q_\Gamma=(q_{i_1},...,q_{i_m})$.
Moreover, in order to distinguish indices for the servers from those for the records, we usually use superscripts as indices for the servers and subscripts for the records. For example, in the below we use $\mathcal{Q}_\theta^{(i)}$ denote the query to the $i$th server when the user wants the $\theta$th record. Throughout the paper we use cursive capital letters to denote random variables, such as $\mathcal{W}, \mathcal{Q}$, etc.

Suppose a database consisting of $M$ records $\mathcal{W}_1,...,\mathcal{W}_M$ is replicated across $N$ servers $\rm{Serv}^{(1)},...,\rm{Serv}^{(N)}$, i.e., each server stores all the $M$ records. In general, we assume the records are independent and of the same size, i.e.
\begin{equation}\label{q1}\begin{split}
&H(\mathcal{W}_1,...,\mathcal{W}_M)=\sum_{i=1}^MH(\mathcal{W}_i)=ML,\mbox{~and}\\
&H(\mathcal{W}_i)=L,~\forall i\in[M]
\end{split}\end{equation}
where $H(\cdot)$ denotes the entropy function. A PIR scheme enables a user to retrieve a record, say $\mathcal{W}_\theta$, for some $\theta\in[M]$, from the $N$ servers while keeping $\theta$ secret from any colluding subsets containing at most $T$ servers. More specifically, a PIR scheme consists of two phases:
\begin{itemize}
  \item {\bf Query phase. }Given an index $\theta\in[M]$ and some random resources $\mathcal{S}$, the user computes ${\rm Que}(\theta,\mathcal{S})=(\mathcal{Q}_\theta^{(1)},...,\mathcal{Q}_\theta^{(N)})$, and sends $\mathcal{Q}_\theta^{(i)}$ to $\rm{Serv}^{(i)}$ for $1\leq i\leq N$. Note that $\mathcal{S}$ and $\theta$ are private information only known by the user, and ${\rm Que}(\cdot,\cdot)$ is the {\it query function} determined by the scheme.
      Define \begin{equation}\label{eqQ}\mathcal{Q}\triangleq\{\mathcal{Q}_\theta^{(i)}|i\in[N], \theta\in[M]\}\cup\{\mathcal{S},\theta\}\;.\end{equation}  It is natural to assume \begin{equation}\label{eq0}
      I(\mathcal{W}_{[M]};\mathcal{Q})=0\end{equation}
      which means the user generates his queries without knowing exact content of the records.
  \item {\bf Response phase. }For $1\leq i\leq N$, the $i$th server $\rm{Serv}^{(i)}$  at receiving $\mathcal{Q}_\theta^{(i)}$, computes ${\rm Ans}^{(i)}(\mathcal{Q}_\theta^{(i)},\mathcal{W}_{[M]})=\mathcal{A}_\theta^{(i)}$ and sends it back to the user, where ${\rm Ans}^{(i)}(\cdot,\cdot)$ is $\rm{Serv}^{(i)}$'s {\it answer function} determined by the scheme. It is easy to see  \begin{equation}\label{eq1}
      H(\mathcal{A}_\theta^{(i)}|\mathcal{W}_{[M]},\mathcal{Q}_\theta^{(i)})=0\end{equation}
\end{itemize}
To design a PIR scheme is to design the functions ${\rm Que}$ and ${\rm Ans}^{(i)}$ such that the following two conditions are satisfied:
\begin{itemize}
\item[(1)]{\it Correctness: } \begin{equation}\label{eq2}
      H(\mathcal{W}_\theta|\mathcal{A}^{[N]}_\theta,\mathcal{Q}^{[N]}_\theta,\mathcal{S},\theta)=0,
      \end{equation} which means the user can definitely recover the record $\mathcal{W}_\theta$ after he gets answers from all the servers. Based on the definition of $\mathcal{Q}$ in (\ref{eqQ}), the correctness conditions is also represented as $H(\mathcal{W}_\theta|\mathcal{A}^{[N]}_\theta,\mathcal{Q})=0$.
\item[(2)]{\it Privacy: } For any $\Gamma\subseteq[N]$ with $|\Gamma|\leq T$,
\begin{equation}\label{eq3}
     I(\theta;\mathcal{Q}^{\Gamma}_\theta,\mathcal{A}^{\Gamma}_\theta,\mathcal{W}_{[M]})=0,
      \end{equation} which means any up to $T$ servers get no information about the index $\theta$. Note that $\mathcal{Q}^{\Gamma}_\theta,\mathcal{A}^{\Gamma}_\theta,\mathcal{W}_{[M]}$ are all the information held by the servers in $\Gamma$.
\end{itemize}
Such a scheme is also called a $T$-private PIR scheme in \cite{Sun&Jafar16:ColludPIR} as it generalizes the PIR scheme in \cite{Sun&Jafar16:CapacityPIR} which stands only for the special case of $T=1$. In this paper, for simplicity we omit the word $T$-private although all of our results work for general values of $T$.

To measure download efficiency of a PIR scheme,
we define its {\it rate} as
$$\mathcal{R}\triangleq\frac{H(\mathcal{W}_\theta)}{\sum_{i=1}^{N}H(\mathcal{A}^{(i)}_\theta)}=\frac{L}{D}\;,$$
where $D$ is the total size of data downloaded by the user from all servers,  i.e. $D=\sum_{i=1}^{N}H(\mathcal{A}^{(i)}_\theta)$. In other words, $\mathcal{R}$ characterizes how many bits of desired information are retrieved per downloaded bit. Furthermore, the supremum of the rate over all PIR schemes working for a given PIR model is called the {\it capacity} of PIR, denoted by $\mathcal{C}_{\mbox{\tiny PIR}}$.

\subsection{Capacity of PIR schemes}\label{sec2b}
For the sake of later use, in this section we briefly reprove capacity for the PIR schemes through a series of key lemmas.
\begin{lemma}\label{lem2}
For any $\theta,\theta^\prime\in[M]$, and for any subset $\Lambda\subseteq[M]$ and $\Gamma\subseteq[N]$ with $|\Gamma|\leq T$,
\begin{equation*}
H(\mathcal{A}^{\Gamma}_{\theta}| \mathcal{W}_\Lambda,\mathcal{Q})=H(\mathcal{A}^{\Gamma}_{\theta^\prime}| \mathcal{W}_\Lambda,\mathcal{Q})\;.\end{equation*}
\end{lemma}
The proof of  Lemma \ref{lem2} can be found in Appendix \ref{appenA}.

\begin{lemma}\label{lem3}
For any subset $\Lambda\subseteq[M]$, for any $\theta\in\Lambda$ and $\theta'\in [M]-\Lambda$,
$$H(\mathcal{A}^{[N]}_{\theta}| \mathcal{W}_\Lambda,\mathcal{Q})\geq \frac{TL}{N}+\frac{T}{N} H(\mathcal{A}^{[N]}_{\theta^\prime}| \mathcal{W}_{\Lambda\cup\{\theta'\}},\mathcal{Q})\;.$$
\end{lemma}
\begin{proof}
Since for any $\Gamma \subseteq[N]$ with  $|\Gamma|=T$,
\begin{equation}\label{eq7}
 H(\mathcal{A}^{[N]}_{\theta}| \mathcal{W}_\Lambda,\mathcal{Q})\geq H(\mathcal{A}^{\Gamma}_{\theta}| \mathcal{W}_\Lambda,\mathcal{Q})\;, \end{equation}   then{\small
 \begin{eqnarray}
 &&H(\mathcal{A}^{[N]}_{\theta}| \mathcal{W}_\Lambda,\mathcal{Q})\notag\\
 &\geq& \frac{1}{\binom{N}{T}}\sum_{\Gamma: \Gamma \subseteq[N],|\Gamma|=T}H(\mathcal{A}^{\Gamma}_{\theta}| \mathcal{W}_\Lambda,\mathcal{Q}) \notag\\
 &\stackrel{(a)}{=}&\frac{1}{\binom{N}{T}}\sum_{\Gamma: \Gamma \subseteq[N],|\Gamma|=T}H(\mathcal{A}^{\Gamma}_{\theta^\prime}| \mathcal{W}_\Lambda,\mathcal{Q}) \notag \\
 &\stackrel{(b)}{\geq}& \frac{T}{N}H(\mathcal{A}^{[N]}_{\theta^\prime}| \mathcal{W}_\Lambda,\mathcal{Q}) \label{eq8}\\
 &=&\frac{T}{N}(H(\mathcal{A}^{[N]}_{\theta^\prime}, \mathcal{W}_{\theta^\prime}| \mathcal{W}_\Lambda,\mathcal{Q})-H(\mathcal{W}_{\theta^\prime}|\mathcal{A}^{[N]}_{\theta^\prime},\mathcal{W}_\Lambda,\mathcal{Q}))\notag\\
 &\stackrel{(c)}{=}&\frac{T}{N}(H(\mathcal{W}_{\theta^\prime}| \mathcal{W}_\Lambda,\mathcal{Q})+ H(\mathcal{A}^{[N]}_{\theta^\prime}| \mathcal{W}_{\Lambda\cup\{\theta'\}},\mathcal{Q})) \notag\\
 &\stackrel{(d)}{=}&\frac{TL}{N} +\frac{T}{N}H(\mathcal{A}^{[N]}_{\theta^\prime}| \mathcal{W}_{\Lambda\cup\{\theta'\}},\mathcal{Q}) \notag
 \end{eqnarray}}
 where  {\small$(a)$} comes from Lemma \ref{lem2}, the inequality {\small$(b)$} follows from the Han's inequality,  {\small$(c)$} is due to the fact $H(\mathcal{W}_{\theta^\prime}|\mathcal{A}^{[N]}_{\theta^\prime},\mathcal{W}_\Lambda,\mathcal{Q})=0$ which is implied from the correctness condition (\ref{eq2}), and {\small$(d)$} is a consequence of the assumption  (\ref{q1}) and (\ref{eq0}).
\end{proof}

\begin{theorem}(\cite{Sun&Jafar16:ColludPIR})\label{thm1}
For $T$-private PIR with $M$ records and $N$ replicated servers, the capacity
\begin{align*}
\mathcal{C}_{\mbox{\tiny PIR}}=\frac{1}{1+\frac{T}{N}+\frac{T^2}{N^2}+\dots+\frac{T^{M-1}}{N^{M-1}}}.
\end{align*}
\end{theorem}

\begin{proof} A general PIR scheme given in \cite{Sun&Jafar16:ColludPIR} proves achievability of this capacity, so here we only reprove for any PIR scheme,
$$\mathcal{R}\leq \frac{1}{1+\frac{T}{N}+\frac{T^2}{N^2}+\dots+\frac{T^{M-1}}{N^{M-1}}}.$$
From the definition, it has
\begin{equation}\label{eq10}\mathcal{R}=\frac{H(\mathcal{W}_\theta)}{\sum^N_{i=1}H(\mathcal{A}_\theta^{(i)})}\leq \frac{L}{H(\mathcal{A}_\theta^{[N]})}\leq \frac{L}{H(\mathcal{A}_\theta^{[N]}|\mathcal{Q})}.\end{equation}Therefore, it is sufficient to prove for any $\theta\in[M]$,
\begin{equation}\label{eq4}
 H(\mathcal{A}_\theta^{[N]}|\mathcal{Q})\geq \sum^{M-1}_{i=0}\frac{T^i}{N^i} L\;.
\end{equation}

 First, we have
 \begin{equation}\label{eq5}
\begin{split}
 L&=H(\mathcal{W}_\theta)\stackrel{(a)}{=}H(\mathcal{W}_\theta|\mathcal{Q})-H(\mathcal{W}_\theta|\mathcal{A}_\theta^{[N]},\mathcal{Q})\\
 &=I(\mathcal{W}_\theta;\mathcal{A}_\theta^{[N]}|\mathcal{Q})\\
 &=H(\mathcal{A}_\theta^{[N]}|\mathcal{Q})-H(\mathcal{A}_\theta^{[N]}|\mathcal{W}_\theta,\mathcal{Q}) \end{split}
\end{equation}
  where {\small$(a)$} comes from (\ref{eq0}) and (\ref{eq2}). Then by Lemma \ref{lem3}, it follows
  $H(\mathcal{A}_{\theta}^{[N]}|\mathcal{W}_\theta,\mathcal{Q})\geq \frac{TL}{N} + \frac{T}{N}(H(\mathcal{A}^{[N]}_{\theta'}| \mathcal{W}_{\{\theta,\theta'\}},\mathcal{Q})$.
 By recursively using Lemma \ref{lem3}, we finally have
{\small
\begin{eqnarray}
H(\mathcal{A}_\theta^{[N]}|\mathcal{W}_\theta,\mathcal{Q})&\geq& \sum^{M-1}_{i=1} \frac{T^i}{N^i} L+
\frac{T^{M-1}}{N^{M-1}} H(\mathcal{A}_{\theta''}^{[N]}|\mathcal{W}_{[M]},\mathcal{Q})\notag\\
&\stackrel{(b)}{=}&\sum^{M-1}_{i=1} \frac{T^i}{N^i} L\label{eq6}
\end{eqnarray}}
where {\small$(b)$} comes from the fact $H(\mathcal{A}_{\theta''}^{[N]}|\mathcal{W}_{[M]},\mathcal{Q})=0$  which is a consequence of (\ref{eq1}). Combining (\ref{eq5}) and (\ref{eq6}), we immediately obtain (\ref{eq4}).
\end{proof}

\section{Lower bound on sub-packetization}
In this section, we prove a lower bound on sub-packetization for all  linear capacity-achieving PIR schemes, i.e., \begin{theorem}\label{thm2} Suppose $M\geq 2, ~N>T\geq1$. For all linear capacity-achieving $T$-private PIR schemes with $M$ records and $N$ replicated servers, its sub-packetization $L$ has a lower bound, i.e.,
\begin{align*}L\geq dn^{M-1}\end{align*}  where $d={\rm gcd}(N,T),n=\frac{N}{d}$.
\end{theorem}
Toward proving this theorem, we first list some useful identities for general capacity-achieving PIR schemes in Section \ref{seca}. Then in Section \ref{secb} we restrict to linear schemes and prove some useful properties. We finish the proof of Theorem \ref{thm2} in Section \ref{secc}.

\subsection{Some identities for capacity-achieving PIR schemes}\label{seca}
Recall that throughout this paper we study the PIR schemes working for $N$ replicated servers with $M$ independent records, and maintaining the $T$-privacy for $1\leq T<N$. For simplicity, we call these schemes just as PIR schemes. Such a scheme is called capacity-achieving if its rate equals the $\mathcal{C}_{\mbox{\tiny PIR}}$ in Theorem \ref{thm1}.

\begin{proposition}\label{pro1} For a capacity-achieving PIR scheme, for any $\theta\in[M]$, it holds
\begin{equation}\label{eq11}
H(\mathcal{A}_\theta^{[N]}|\mathcal{Q})=\sum_{i=1}^{N}H(\mathcal{A}_\theta^{(i)})=D
\end{equation}
\begin{equation}\label{eq12}
H(\mathcal{A}_\theta^{[N]}|\mathcal{W}_\theta,\mathcal{Q})=D-L.
\end{equation}
\end{proposition}
\begin{proof}
From (\ref{eq10}) in the proof of Theorem \ref{thm1},  it has
\begin{align*}
\mathcal{R}&=\frac{H(\mathcal{W}_\theta)}{\sum^N_{i=1}H(\mathcal{A}_\theta^{(i)})}=\frac{L}{D}\leq \frac{L}{H(\mathcal{A}_\theta^{[N]}|\mathcal{Q})} \\
&\leq \frac{1}{1+\frac{T}{N}+\frac{T^2}{N^2}+\dots+\frac{T^{M-1}}{N^{M-1}}}.
\end{align*}
for any PIR scheme. Particularly, for any  capacity-achieving PIR scheme, the equalities above hold simultaneously. Therefore, we obtain
$H(\mathcal{A}_\theta^{[N]}|\mathcal{Q})=\sum_{i=1}^{N}H(\mathcal{A}_\theta^{(i)})=D$.

The identity (\ref{eq12}) follows from (\ref{eq5}) and (\ref{eq11}) immediately.
\end{proof}

\begin{proposition}\label{pro2}
For any capacity-achieving PIR scheme, for any $\theta\in[M]$, any $\Lambda\subseteq[M]$, and any $\Gamma\subseteq[N]$ with $|\Gamma|=T$, it holds
 \begin{equation*}H(\mathcal{A}_\theta^{\Gamma}|\mathcal{W}_{\Lambda},\mathcal{Q})= \begin{cases}
 H(\mathcal{A}_\theta^{[N]}|\mathcal{W}_{\Lambda},\mathcal{Q}),& if~\theta\in \Lambda\\
  \frac{T}{N}H(A_\theta^{[N]}|\mathcal{W}_{\Lambda},\mathcal{Q}), & if ~\theta\notin\Lambda
\end{cases}\end{equation*}
\end{proposition}

\begin{proof}
For any capacity-achieving PIR scheme, the inequalities in the proof of Lemma \ref{lem3} must hold with  equalities simultaneously.  Specifically, from (\ref{eq7}) we can see for any $\theta \in \Lambda$,
 \begin{equation}\label{eq14}
 H(\mathcal{A}_\theta^{\Gamma}|\mathcal{W}_{\Lambda},\mathcal{Q})=H(\mathcal{A}_\theta^{[N]}|\mathcal{W}_{\Lambda},\mathcal{Q})\;,
 \end{equation}
 and from (\ref{eq8}) we can see for any $\theta' \in[M]-\Lambda$,
 \begin{eqnarray*}
 &&\frac{T}{N}H(\mathcal{A}_{\theta'}^{[N]}|\mathcal{W}_{\Lambda},\mathcal{Q})\\
 &=&\frac{1}{\binom{N}{T}}\sum_{\Gamma':\Gamma'\subseteq[N], |\Gamma'|=T}H(\mathcal{A}_{\theta'}^{\Gamma'}|\mathcal{W}_{\Lambda},\mathcal{Q})\\
 &\stackrel{(a)}{=}&\frac{1}{\binom{N}{T}}\sum_{\Gamma':\Gamma'\subseteq[N], |\Gamma'|=T}H(\mathcal{A}_{\theta}^{\Gamma'}|\mathcal{W}_{\Lambda},\mathcal{Q})\\
 &\stackrel{(b)}{=}&H(\mathcal{A}_{\theta}^{[N]}|\mathcal{W}_{\Lambda},\mathcal{Q})\\
 &\stackrel{(c)}{=}&H(\mathcal{A}_{\theta}^{\Gamma}|\mathcal{W}_{\Lambda},\mathcal{Q})\\
 &\stackrel{(d)}{=}&H(\mathcal{A}_{\theta'}^{\Gamma}|\mathcal{W}_{\Lambda},\mathcal{Q})
 \end{eqnarray*}
   where {\small$(a)$} and {\small$(d)$} come from Lemma \ref{lem2},  while {\small$(b)$} and {\small$(c)$} follow from (\ref{eq14}).
\end{proof}

\subsection{Properties of linear capacity-achieving PIR schemes}\label{secb}
We first give formal definitions of linear PIR schemes and sub-packetization.
\begin{definition}\label{def}
Suppose each record is expressed as an $L$-dimensional vector over $\mathbb{F}_q$, i.e., $W_i\in\mathbb{F}_q^L$ for $1\leq i\leq M$. A PIR scheme is called linear over $\mathbb{F}_q$ if for retrieving any record $W_\theta$, $\theta\in[M]$, the answers provided by $\rm{Serv}^{(j)}$, $1\leq j\leq N$, is {\small$$A_\theta^{(j)}=\sum_{i=1}^MW_iQ_{\theta,i}^{(j)}\in\mathbb{F}_q^{~\gamma_j}\;,$$} where $Q_{\theta,i}^{(j)}$ is an $L\times\gamma_j$ matrix over $\mathbb{F}_q$, $1\leq i\leq M$. Equivalently, for a linear PIR scheme, it has
{\small\begin{eqnarray}
  &&(A_\theta^{(1)},A_\theta^{(2)},...,A_\theta^{(N)})\nonumber\\
  &=&(W_1,W_2,...,W_M)\begin{pmatrix}
  Q_{\theta,1}^{(1)}&Q_{\theta,1}^{(2)}&\cdots&Q_{\theta,1}^{(N)}\\
  Q_{\theta,2}^{(1)}&Q_{\theta,2}^{(2)}&\cdots&Q_{\theta,2}^{(N)}\\
  \vdots&\vdots&\vdots&\vdots\\
  Q_{\theta,M}^{(1)}&Q_{\theta,M}^{(2)}&\cdots&Q_{\theta,M}^{(N)}
  \end{pmatrix}
\end{eqnarray}}
and the matrix $(Q_{\theta,i}^{(j)})_{\substack{ 1\leq i\leq M\\1\leq j\leq N}}$ defines the queries for retrieving $\theta$.
Moreover, $L$ is called {\it sub-packetization} of the linear PIR scheme.
\end{definition}

\begin{remark}
Note that a linear PIR scheme is implemented by linear operations over a finite field $\mathbb{F}_q$. That is, a symbol in $\mathbb{F}_q$ is an atomic unit in implementation of the scheme. Since $L$ denotes the number of symbols contained in each record that are necessary for the scheme implementation, we call $L$ as sub-packetization of the scheme. It can be seen small values of $L$ imply less operations involved in realization, and thus a simpler scheme.
\end{remark}

\begin{remark}\label{re2}
Recall in Section \ref{sec2} and (\ref{q1}), $L$ has been used to denote the entropy of each record, i.e. $H(\mathcal{W}_i)=L$ for $1\leq i\leq M$. For linear PIR schemes, we represent each record as an $L$-dimensional vector over $\mathbb{F}_q$. Because it is natural to assume that each coordinate of a record is independently and uniformly distributed in $\mathbb{F}_q$, by choosing a proper logarithm in the entropy function each record in the $L$-dimensional vector-representation also has entropy $L$.
Moreover, for any random variable in the vector-representation, its entropy actually equals the number of independent symbols in $\mathbb{F}_q$ contained in the vector.
\end{remark}

Note that in Section \ref{sec2} we regard each variable as a random variable, while in this section these variables are represented as vectors or matrices. A key point for proving Theorem \ref{thm2} is to establish connections between entropy of random variables and ranks of the corresponding linear objects. In Remark \ref{re2} we explain the basic connection, and in the next lemma we give another important connection.

\begin{lemma}\label{lem4}For a linear capacity-achieving  PIR scheme, for any $\theta\in[M]$ and any nonempty sets $\Gamma\subseteq[N],\Lambda\subseteq[M]$, it has
\begin{equation}\label{eq15}
H(\mathcal{A}_\theta^{\Gamma}|\mathcal{W}_{\Lambda},\mathcal{Q})={\rm rank}(Q^{\Gamma}_{\theta,[M]-\Lambda}).
\end{equation}
\end{lemma}
\begin{proof}
Firstly, from Definition \ref{def}, we have
\begin{equation}\label{eqran}A_\theta^{\Gamma}=\sum_{i\in\Lambda}W_iQ^{\Gamma}_{\theta,i}+\sum_{i\in[M]-\Lambda}W_iQ^{\Gamma}_{\theta,i}\;.\end{equation}
On the other hand, the left side of (\ref{eq15}) is the entropy of $\mathcal{A}_\theta^{\Gamma}$  on the condition of $\mathcal{W}_{\Lambda}$ and $\mathcal{Q}$. Thus based on the correspondence between random variables and linear objects, this lemma actually requires to compute the entropy of $A_\theta^{\Gamma}$ for fixed $W_{\Lambda}$ and $Q\triangleq(Q_{\theta,i}^{(j)})_{\substack{ 1\leq i\leq M\\1\leq j\leq N}}$.

Since $W_{\Lambda}$ and $Q$ are fixed, from (\ref{eqran}) we can see randomness of $A_\theta^{\Gamma}$ only depends on $\sum_{i\in[M]-\Lambda}W_iQ^{\Gamma}_{\theta,i}$. That is, $A_\theta^{\Gamma}$ and $\sum_{i\in[M]-\Lambda}W_iQ^{\Gamma}_{\theta,i}$ have the identical distribution. While the latter is linear combinations of independent symbols $W_{i,s}$, $i\in[M]-\Lambda$ and $1\leq s\leq L$. Suppose ${\rm rank}(Q^{\Gamma}_{\theta,[M]-\Lambda})=r$, then $Q^{\Gamma}_{\theta,[M]-\Lambda}$ has $r$ linearly independent columns, corresponding to $r$ coordinates of $\sum_{i\in[M]-\Lambda}W_iQ^{\Gamma}_{\theta,i}$. It can be seen that these $r$ coordinates completely determine the vector $\sum_{i\in[M]-\Lambda}W_iQ^{\Gamma}_{\theta,i}$. Moreover, each of the $r$ coordinates is independently and uniformly distributed in $\mathbb{F}_q$. Because the entropy is measured by the number of independent symbols in $\mathbb{F}_q$, the lemma follows.
\end{proof}

Based on the connections stated in Remark \ref{re2}  and Lemma \ref{lem4}, we can further derive some properties of the linear objects in linear PIR schemes.

\begin{proposition}\label{pro4} For a linear capacity-achieving  PIR scheme, for any $\theta\in[M]$,
it holds
\begin{equation*}
{\rm rank}(Q^{[N]}_{\theta,\theta})=L\;.
\end{equation*}
\end{proposition}
\begin{proof}
For simplicity, denote $\bar{\theta}\triangleq[M]-\theta$.
By Lemma \ref{lem4} we have ${\rm rank}(Q^{[N]}_{\theta,\theta})=H(\mathcal{A}_\theta^{[N]}|\mathcal{W}_{\bar{\theta}},\mathcal{Q})$. Then
\begin{eqnarray*}
&&H(\mathcal{A}_\theta^{[N]}|\mathcal{W}_{\bar{\theta}},\mathcal{Q})\\
&\stackrel{(a)}{=}&H(\mathcal{A}_\theta^{[N]}|\mathcal{W}_{\bar{\theta}},\mathcal{Q})-H(\mathcal{A}_\theta^{[N]}|\mathcal{W}_{[M]},\mathcal{Q})\\
&=&I(\mathcal{A}_\theta^{[N]};\mathcal{W}_{\theta}|\mathcal{W}_{\bar{\theta}},\mathcal{Q})\\
&=&H(\mathcal{W}_{\theta}|\mathcal{W}_{\bar{\theta}},\mathcal{Q})-H(\mathcal{W}_{\theta}|\mathcal{A}_\theta^{[N]},\mathcal{W}_{\bar{\theta}},\mathcal{Q})\\
&\stackrel{(b)}{=}&H(\mathcal{W}_{\theta}|\mathcal{W}_{\bar{\theta}},\mathcal{Q})\\
&\stackrel{(c)}{=}&L
\end{eqnarray*}
where {\small$(a)$} is due to (\ref{eq1}), {\small$(b)$} comes from the correctness condition (\ref{eq2}), and {\small$(c)$} is from the assumptions (\ref{q1}) and (\ref{eq0}).
\end{proof}

\begin{proposition}
\label{lem6} For a linear capacity-achieving PIR scheme, for any $\theta\in[M]$ and any $\Gamma\subseteq[N]$ with $|\Gamma|=T$,
it holds
\begin{equation}\label{eq18}
{\rm rank}(Q^{\Gamma}_{\theta,\theta})=\frac{TL}{N},
\end{equation}
\begin{equation}\label{eq19}
{\rm rank}(Q^{\Gamma}_{\theta,\bar{\theta}})=D-L.
\end{equation}
\end{proposition}

\begin{proof} From Lemma \ref{lem4}, we have
\begin{eqnarray*}
{\rm rank}(Q^{\Gamma}_{\theta,\theta})&=&H(\mathcal{A}_\theta^{\Gamma}|\mathcal{W}_{\bar{\theta}},\mathcal{Q})\\
&\stackrel{(a)}{=}&\frac{T}{N}H(\mathcal{A}_\theta^{[N]}|\mathcal{W}_{\bar{\theta}},\mathcal{Q})\\
&\stackrel{(b)}{=}&\frac{T}{N}{\rm rank}(Q^{[N]}_{\theta,\theta})\\&\stackrel{(c)}{=}&\frac{TL}{N}
\end{eqnarray*}
where {\small$(a)$} is from Proposition \ref{pro2}, {\small$(b)$} is due to Lemma \ref{lem4} and {\small$(c)$} comes from Proposition \ref{pro4}.
Similarly, by successively using Lemma \ref{lem4}, Proposition \ref{pro2} and (\ref{eq12})  we have
\begin{eqnarray*}
{\rm rank}(Q^{\Gamma}_{\theta,\bar{\theta}})&=&H(\mathcal{A}_\theta^{\Gamma}|\mathcal{W}_{\theta},\mathcal{Q})\\
&=&H(\mathcal{A}_\theta^{[N]}|\mathcal{W}_{\theta},\mathcal{Q})\\
&=&D-L .
\end{eqnarray*}
\end{proof}

\subsection{Proof of Theorem \ref{thm2}}\label{secc}
Some basic lemmas are stated below before proving Theorem \ref{thm2}. Proofs of these lemmas can be found in Appendix \ref{appenB}.

\begin{lemma}\label{lem5}
Let $a,b,m\in\mathbb{N}$. Suppose $d_1={\rm gcd}(a,b),d_2={\rm gcd}(a^m,\sum^{m}_{i=0}a^{m-i}b^i)$, then $d_2=d_1^m$.
\end{lemma}

\begin{lemma}
\label{thm3}Let $V_1,...,V_N$ be $N$ linear subspaces of $\mathbb{F}^{L}$ and denote $V=\sum^N_{i=1}V_i$. Suppose $T\in\mathbb{N}$ and $1\leq T<N$. If for any subset $\Gamma\subseteq[N]$ with $|\Gamma|=T$ it holds $\dim(\sum_{i\in\Gamma}V_i)=\frac{T}{N}\dim V$, then
$N|\dim V$ and $\dim(V_i)=\frac{1}{N}\dim V$ for $1\leq i\leq N$. Moreover ,
 \begin{equation}\label{eq23}
V=\oplus^N_{i=1}V_i,
\end{equation}
\end{lemma}

Next we prove Theorem \ref{thm2}.
\begin{proof}We finish the proof in three steps:

{\it(1)  Prove $L$ and $D$ have specific forms, i.e.,
$L=\mu n^{M-1}$ and $D=\mu \sum^{M-1}_{i=0}n^{M-1-i}t^i$ for some $\mu\in\mathbb{N}$.}

For capacity-achieving schemes, we have
\begin{align}\label{eqrate}
\mathcal{R}=\frac{L}{D}&=\frac{1}{1+\frac{t}{n}+\dots+\frac{t^{M-1}}{n^{M-1}}}=\frac{n^{M-1}}{\sum^{M-1}_{i=0}n^{M-1-i}t^i}.
\end{align}
Based on the connection stated in Remark \ref{re2}  both $L$ and $D$ are integers, thus (\ref{eqrate}) implies $n^{M-1}|L\sum^{M-1}_{i=0}n^{M-1-i}t^i$. Note that ${\rm gcd}(n,t)=1$, so by Lemma \ref{lem5} it has
${\rm gcd}(n^{M-1},\sum^{M-1}_{i=0}n^{M-1-i}t^i)=1$. As a result, we have $n^{M-1}|L$. That is, there exists an integer $\mu$ such that $L=\mu n^{M-1}$. Combining with (\ref{eqrate}) it follows $D=\mu \sum^{M-1}_{i=0}n^{M-1-i}t^i$.

\vspace{6pt}

{\it(2)  Prove $N\mid L$.}

For any $i\in [N]$, denote $V_i=\textsl{Cspan}(Q^{(i)}_{\theta,\theta})$ where $\textsl{Cspan}(\cdot)$ denotes the linear space spanned by all columns of the matrix. Thus from Proposition \ref{pro4} it has ${\rm dim}(\sum_{i=1}^NV_i)={\rm rank}(Q^{[N]}_{\theta,\theta})=L$, and from (\ref{eq18}) it has ${\rm dim}(\sum_{i\in\Gamma}V_i)={\rm rank}(Q^{\Gamma}_{\theta,\theta})=\frac{T}{N}L$ for any $\Gamma\subseteq[N]$ with $|\Gamma|=T$. Therefore, by Lemma \ref{thm3} it follows $N\mid L$.

Moreover, it has $V=\oplus_{i=1}^NV_i$ which implies ${\rm rank} (Q^{\Gamma}_{\theta,\theta})=\sum_{i\in\Gamma}{\rm rank} (Q^{i}_{\theta,\theta})$. By Lemma \ref{lem4} we have
 \begin{equation}\label{eq24}
 H(\mathcal{A}^{\Gamma}_\theta|\mathcal{W}_{\bar{\theta}},\mathcal{Q})
 =\sum_{i\in\Gamma}H(\mathcal{A}^{(i)}_\theta|\mathcal{W}_{\bar{\theta}},\mathcal{Q})\;.
\end{equation}

\vspace{6pt}

{\it(3) Prove $T\mid D-L$}

For any $\theta'\neq\theta$, it can be seen that
\begin{eqnarray}
{\rm rank}(Q^{\Gamma}_{\theta,\theta'})&\stackrel{(a)}{=}& H(\mathcal{A}^{\Gamma}_\theta|\mathcal{W}_{\overline{\theta'}},\mathcal{Q})\notag\\
&\stackrel{(b)}{=}&H(\mathcal{A}^{\Gamma}_{\theta'}|\mathcal{W}_{\overline{\theta'}},\mathcal{Q})\notag\\
 &\stackrel{(c)}{=}&\sum_{i\in\Gamma}H(\mathcal{A}^{(i)}_{\theta'}|\mathcal{W}_{\overline{\theta'}},\mathcal{Q})\notag\\
 &\stackrel{(d)}{=}&\sum_{i\in\Gamma}H(\mathcal{A}^{(i)}_\theta|\mathcal{W}_{\overline{\theta'}},\mathcal{Q})\notag\\
 &\stackrel{(e)}{=}&\sum_{i\in\Gamma}{\rm rank} (Q^{(i)}_{\theta,\theta'})\label{eq25}
\end{eqnarray}
where {\small(a)} and {\small (e)} hold because of Lemma \ref{lem4}, {\small (b)} and {\small (d)} is due to Lemma \ref{lem2}, and {\small (c)} comes from (\ref{eq24}).

Because (\ref{eq25}) holds for all $\theta'\in\bar{\theta}$, it follows ${\rm rank}(Q^{\Gamma}_{\theta,\bar{\theta}})=\sum_{i\in\Gamma}{\rm rank} (Q^{(i)}_{\theta,\bar{\theta}})$. From (\ref{eq19}) we know ${\rm rank}(Q^{\Gamma}_{\theta,\bar{\theta}})=D-L$. Therefore, for any set $\Gamma\subseteq[N]$ with $|\Gamma|=T$, it holds $\sum_{i\in\Gamma}{\rm rank} (Q^{(i)}_{\theta,\bar{\theta}})=D-L$. It can be derived that ${\rm rank} (Q^{(i)}_{\theta,\bar{\theta}})={\rm rank} (Q^{(j)}_{\theta,\bar{\theta}})$ for any $i,j\in[N]$ and $T\mid D-L$.

\vspace{6pt}

Finally, from $N\mid L$ and $L=\mu n^{M-1}$ we have $d\mid \mu n^{M-2}$. Similarly, from $T\mid D-L$ and $D-L=\mu \sum^{M-1}_{i=1}n^{M-1-i}t^i$ we have $d\mid \mu \sum^{M-2}_{i=0}n^{M-2-i}t^i$. Therefore,
$$d\mid{\rm gcd}(\mu n^{M-2},\mu \sum^{M-2}_{i=0}n^{M-2-i}t^i)\;.$$
However, because ${\rm gcd}(n,t)=1$ we know from Lemma \ref{lem5} that ${\rm gcd}(n^{M-2},\sum^{M-2}_{i=0}n^{M-2-i}t^i)=1$ . Hence we have $d\mid \mu$ and thus $\mu\geq d$. As a result, $L=\mu n^{M-1}\geq dn^{M-1}$.
\end{proof}

\section{The PIR scheme with $L=dn^{M-1}$}
In this section, for $M\geq 2$ and $1\leq T<N$, we present a linear capacity-achieving PIR scheme with sub-packetization $L=dn^{M-1}$.
\subsection{Examples}
We begin with a simple example.
\begin{example}\label{eg1} Suppose $M=2$, $N=3$ and $T=2$. In this case the sub-packetization of our scheme is $L=d(\frac{N}{d})^{M-1}=3$, so we regard each record as a $3$-dimensional vector over a finite field $\mathbb{F}_{q}$, i.e., $W_1,W_2\in\mathbb{F}_{q}^{~3}$.\footnote{For records of large size, they can be divided into stripes in $\mathbb{F}_{q}^{~3}$, i.e., $W_i\in\mathbb{F}_{q^\ell}^{~3}$ for $\ell$ stripes. The scheme works identically on all the stripes in parallel. } WLOG, suppose the user wants $W_1$. The scheme is described in two parts:

{\it Part $1$: Mixing-Expanding.}

Let $S_1,S_2\in\mathbb{F}_q^{3\times3}$ be two matrices chosen by the user independently and uniformly from all $3\times3$ invertible matrices over $\mathbb{F}_q$. Actually, $S_1$ and $S_2$ are the random resources privately held by the user. Let $G_{[3,2]}\in\mathbb{F}_{q}^{2\times3}$ be a generator matrix of a $[3,2]$ MDS code over $\mathbb{F}_q$. Then, define
   \begin{align}
   &(a_1,a_2,a_3)=W_1S_1 \nonumber\\
   &(b_1,b_2,b_3)=W_2S_2[:,(1:2)]G_{[3,2]}\label{eqb}
   \end{align}
 where $S_2[:,(1:2)]$  denotes the $3\times 2$ matrix formed by the first $2$ columns of $S_2$.
 Note that $W_1S_1$ induces an invertible transformation of $W_1$ and we call this a {\it mixing} process. Similarly, $W_2S_2[:,(1:2)]$ generates a mixing vector from $W_2$ with only $2$ coordinates drawn from $3$ coordinates. Then by multiplying $G_{[3,2]}$ it expands a $2$-dimensional vector into a $3$-dimensional vector. The MDS property ensures that any $2$ coordinates of $(b_1,b_2,b_3)$ suffice to recover the entire vector.

 {\it Part $2$: Combining.}

 Since the $a_i$'s contain information of the desired record $W_1$, we call them {\it desired symbols}. Correspondingly, the $b_i$'s are called {\it undesired symbols}. For a combination $a_i+b_j$ we call it a {\it mixed symbol}. At queries of the user, the answers given by all the servers are composed of the three kinds of symbols. Specifically, the answers are formed by iteratively applying the following  two steps:
 \begin{itemize}
 \item[(a1)]\emph{Enforcing record symmetry within the query to each server.}
 \item[(a2)] \emph{Combining undesired symbols with new desired symbols.}
 \end{itemize}
 We  explain these two steps through the example. As displayed in Fig.\ref{fg1}, the user first asks for $a_1$ and $a_2$ respectively from $\rm{Serv}^{(1)}$ and $\rm{Serv}^{(2)}$. These are both desired symbols. Then applying Step (a1), he simultaneously asks for $b_1$ and $b_2$ to enforce symmetry with respect to both  records within the query to each server. Note that $b_1$  and $b_2$ are undesired symbols. A high-rate PIR scheme will try to use these undesired symbols decode more desired symbols. Thus applying Step (a2) the user also asks for $a_3+b_3$ from $\rm{Serv}^{(3)}$.

 \begin{figure}[ht]
\centering
\begin{tikzpicture}[scale=2]

\node at (-0.3,0){\footnotesize
\begin{tabular}{|c|c|c|}
\hline $\rm{Serv}^{(1)}$ & $\rm{Serv}^{(2)}$ & $\rm{Serv}^{(3)}$\\\hline
$a_1$&$a_2$&\\\hline
\end{tabular}};
\draw[thick,->,>=stealth] (0.78,0) -- (1.18,0);
\node at (0.95,0.1){\scriptsize(a1)};
\draw[thick,->,>=stealth] (1.18,-0.1) -- (0.83,-0.32);
\node at (0.96,-0.18){\scriptsize(a2)};

\node at (2.25,0){\footnotesize\begin{tabular}{|c|c|c|}
\hline$\rm{Serv}^{(1)}$ & $\rm{Serv}^{(2)}$ & $\rm{Serv}^{(3)}$\\\hline
$a_1,b_1$&$a_2,b_2$&\\\hline
\end{tabular}};

\node at (1.05,-0.62){\footnotesize\begin{tabular}{|c|c|c|}
\hline$\rm{Serv}^{(1)}$ & $\rm{Serv}^{(2)}$ & $\rm{Serv}^{(3)}$\\\hline
$a_1,b_1$&$a_2,b_2$& \\&&$a_3+b_3$\\\hline
\end{tabular}};
\end{tikzpicture}
\caption{Query sequence in the $(M=2,N=3,T=2)$ PIR scheme.}
\label{fg1}
\end{figure}
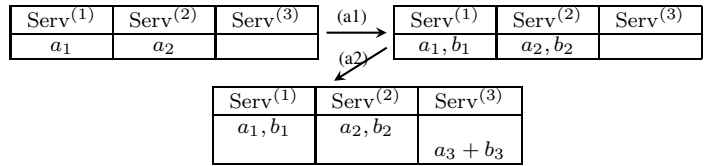

Now we show the scheme satisfies the correctness condition and the privacy condition. Recall that $(b_1,b_2,b_3)$ is a codeword in a $[3,2]$ MDS code. Thus the desired symbol $a_3$ can be decoded from the mixed symbol $a_3+b_3$ by using $b_1$ and $b_2$. The correctness condition follows immediately.

As to the privacy condition, we need to show that for any up to $T=2$ colluding servers, their query sequence for retrieving $W_1$, which is actually a random variable induced by the random matrices $S_1$ and $S_2$, has the same distribution as the query sequence for retrieving $W_2$. Actually, if the user wants $W_2$, the scheme goes the same way except for exchanging the generation processes of $a_i$'s and $b_i$'s in (\ref{eqb}), i.e.,
\begin{align}
   &(a_1,a_2,a_3)=W_1S'_1[:,(1:2)]G_{[3,2]} \nonumber\\
   &(b_1,b_2,b_3)=W_2S'_2.\label{eqa}
   \end{align}
For any two colluding servers, say ${\rm Serv}^{(1)}$ and ${\rm Serv}^{(3)}$, their query sequence is $\{a_1, b_1, a_3+b_3\}$. For every choice of $S_1$ and $S_2$ in (\ref{eqb}), one can correspondingly find $S'_1$ and $S'_2$ in (\ref{eqa}) such that they produce the same  sequence $\{a_1, a_3, b_1, b_3\}$, i.e.,
\begin{align}\label{eqab}(a_1,a_3)=W_1S_1[:,(1,3)]=W_1S'_1[:,(1:2)]G_{[3,2]}[:,(1,3)]\nonumber\\
(b_1,b_3)=W_2S_2[:,(1:2)]G_{[3,2]}[:,(1,3)]=W_2S'_2[:,(1,3)]\end{align}
where $G_{[3,2]}[:,(1,3)]$ is a $2\times2$ matrix formed by the first and the third columns of $G_{[3,2]}$. Because $G_{[3,2]}[:,(1,3)]$ is an invertible matrix, the one-to-one correspondence between $S_i$ and $S_i'$ satisfying (\ref{eqab}) can be easily established, implying the identical distribution of $\{a_1, b_1, a_3+b_3\}$ for retrieving $W_1$ and for
retrieving $W_2$. Similar observations hold for all colluding subset containing at most $2$ servers, thus the privacy condition follows.

Finally, it is easy to see the desired record consists of $3$ symbols while the answers totally contain $5$ symbols, so the scheme has rate $\frac{3}{5}$ attaining the
capacity for this case.
\end{example}

Next we give some remarks to reveal more design ideas behind the scheme construction.

\begin{remark}
In the mixing-expanding part, the desired record goes through an invertible transformation and generates a symbol vector consisting of $L$ independent symbols, while each of the undesired records abandons a fraction of symbols and then generates an $L$-dimensional symbol vector with redundancy through MDS codes.
\end{remark}

\begin{remark}
In the combining part, the answers provided by each server are generated by multiplying the record vector by  a {\it combining matrix} which is fixed for given values of $M,N,T$, and independent of the index of the desired record.

As in Example \ref{eg1}, the answers of  the first server ${\rm Serv}^{(1)}$ is $(a_1,b_1)$. Obviously,
\begin{equation}\label{eqc}(a_1,b_1)=(a_{[1:3]},b_{[1:3]})\begin{pmatrix}{\bf e}_1^\tau&\bf 0\\\bf 0&{\bf e}_1^\tau\end{pmatrix}\end{equation} where ${\bf e}_1=(1,0,0)$. Thus the combining matrix for ${\rm Serv}^{(1)}$ is ${\footnotesize \begin{pmatrix}{\bf e}_1^\tau&\bf 0\\\bf 0&{\bf e}_1^\tau\end{pmatrix}}$. Similarly, one can write down the combining matrices for ${\rm Serv}^{(2)}$ and ${\rm Serv}^{(3)}$ from the tables in Fig. \ref{fg1}.

Furthermore, by (\ref{eqb}) and (\ref{eqc}) the answers of ${\rm Serv}^{(1)}$ for retrieving $W_1$ can be written as
{\small\begin{eqnarray*}(a_1,b_1)&=&(W_1,W_2)\begin{pmatrix}S_1&\\&S_2[:,(1:2)]G_{[3,2]}\end{pmatrix}\begin{pmatrix}{\bf e}_1^\tau&\bf 0\\\bf 0&{\bf e}_1^\tau\end{pmatrix}\\&\triangleq&(W_1,W_2)Q\end{eqnarray*}}
where $Q\in\mathbb{F}_q^{6\times 2}$ is actually the query sent from the user to ${\rm Serv}^{(1)}$ for retrieving $W_1$. It can be seen that $Q$ is a random variable induced by the random resources $S_1$ and $S_2$, and in (\ref{eqab}) we actually proves $Q$ has the identical distribution no matter which record is desired. In our scheme description, we usually omit the complete form of query matrices as they can be directly obtained by putting the mixing-expanding part and the combining part together.

\end{remark}


Let us give one more example to help further understand the scheme construction.
\begin{example}\label{eg2}Suppose $M=3,N=3,T=2$. In this case, the sub-packetization of our scheme is $L=d(\frac{N}{d})^{M-1}=9$, so the records are denoted by $W_1,W_2,W_3\in\mathbb{F}_{q}^{9}$. WLOG, suppose the user wants $W_1$.

In the {\it mixing-expanding} part, the user privately chooses matrices $S_1,S_2,S_3$ randomly and uniformly from all $9\times9$ invertible matrices over $\mathbb{F}_{q}$. Let $G_{[n,k]}\in\mathbb{F}_{q}^{k\times n}$ denote a generator matrix of an $[n,k]$ MDS code which is publicly kown. Then, define
\begin{align*}
   &a_{[1:9]}=W_1S_1\\
   &b_{[1:9]}=W_2S_2[:,{\tiny(1:6)}]G\\
   &c_{[1:9]}=W_3S_3[:,(1:6)]G
\end{align*}
where $G={\footnotesize\begin{pmatrix}
G_{[6,4]}&0\\0&G_{[3,2]}
\end{pmatrix}}$. Actually, the $6$ symbols drawn from $W_2$, i.e., $W_2S_2[:,(1:6)]$, by multiplying $G$ are expanded in two parts separately, that is, the first $4$ symbols are expanded into $6$ symbols by multiplying $G_{[6,4]}$ while the latter $2$ symbols are expanded into $3$ symbols by multiplying $G_{[3,2]}$. The same expanding process goes with $W_3$.

In the {\it Combining} part, the generating process of  queries to all serves by iteratively applying Step (a1) and (a2) is displayed in Fig.\ref{fg2}.
\begin{figure}[ht]
\centering
\begin{tikzpicture}[scale=2]

\node at (-0.2,0){\footnotesize
\begin{tabular}{|c|c|c|}
\hline $\rm{Serv}^{(1)}$ & $\rm{Serv}^{(2)}$ & $\rm{Serv}^{(3)}$\\\hline
$a_1$&$a_2$&$a_3,a_4$\\\hline
\end{tabular}};

\node at (2.3,-0.15){\footnotesize\begin{tabular}{|c|c|c|}
\hline$\rm{Serv}^{(1)}$ & $\rm{Serv}^{(2)}$ & $\rm{Serv}^{(3)}$\\\hline
\makecell{$a_1$\\$b_1$\\$c_1$}&\makecell{$a_2$\\$b_2$\\$c_2$}&\makecell{$a_3,a_4$\\$b_3,b_4$\\$c_3,c_4$}\\
\hline
\end{tabular}};

\node at (-0.2,-1.05){\footnotesize\begin{tabular}{|c|c|c|}
\hline$\rm{Serv}^{(1)}$ & $\rm{Serv}^{(2)}$ & $\rm{Serv}^{(3)}$\\\hline
\makecell{$a_1$\\$b_1$\\$c_1$}&\makecell{$a_2$\\$b_2$\\$c_2$}&\makecell{$a_3,a_4$\\$b_3,b_4$\\$c_3,c_4$}\\$a_5+b_5$&$a_6+b_6$&\\
$a_7+c_5$&$a_8+c_6$&\\
\hline
\end{tabular}};

\node at (2.3,-1.2){\footnotesize\begin{tabular}{|c|c|c|}
\hline$\rm{Serv}^{(1)}$ & $\rm{Serv}^{(2)}$ & $\rm{Serv}^{(3)}$\\\hline
\makecell{$a_1$\\$b_1$\\$c_1$}&\makecell{$a_2$\\$b_2$\\$c_2$}&\makecell{$a_3,a_4$\\$b_3,b_4$\\$c_3,c_4$}\\$a_5+b_5$&$a_6+b_6$&\\
$a_7+c_5$&$a_8+c_6$&\\$b_7+c_7$&$b_8+c_8$&\\
\hline
\end{tabular}};

\node at (1.0,-2.55){\footnotesize\begin{tabular}{|c|c|c|}
\hline$\rm{Serv}^{(1)}$ & $\rm{Serv}^{(2)}$ & $\rm{Serv}^{(3)}$\\\hline
\makecell{$a_1$\\$b_1$\\$c_1$}&\makecell{$a_2$\\$b_2$\\$c_2$}&\makecell{$a_3,a_4$\\$b_3,b_4$\\$c_3,c_4$}\\$a_5+b_5$&$a_6+b_6$&\\
$a_7+c_5$&$a_8+c_6$&\\$b_7+c_7$&$b_8+c_8$&\\&&$a_9+b_9+c_9$\\
\hline
\end{tabular}};

\draw[thick,->,>=stealth] (0.85,0) -- (1.25,0);
\draw[thick,->,>=stealth] (1.25,-0.2) -- (0.9,-0.5);
\draw[thick,->,>=stealth] (0.85,-1.15) -- (1.25,-1.15);
\draw[thick,->,>=stealth] (1.25,-1.5) -- (0.9,-1.85);

\node at (1.05,0.1){\scriptsize(a1)};
\node at (1.05,-0.3){\scriptsize(a2)};
\node at (1.05,-1.05){\scriptsize(a1)};
\node at (1.05,-1.62){\scriptsize(a2)};

\end{tikzpicture}
\caption{Query sequence in the $(M=3,N=3,T=2)$ PIR scheme.}
\label{fg2}
\end{figure}
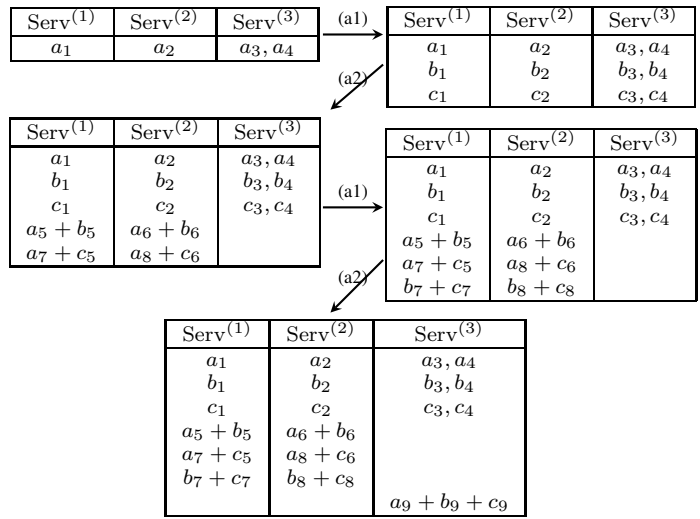

The correctness condition holds obviously. Note that $b_7+c_7, b_8+c_8$ are undesired symbols and they produce $b_9+c_9$ in the expanding process through a $[3,2]$ MDS code. Therefore, the user can obtain $a_9$ from $a_9+b_9+c_9$. The privacy condition can be verified similarly as in Example \ref{eg1}. The rate of this scheme is $\frac{9}{19}$ which matches the capacity for this case.
\end{example}

The above two examples provide an overview of the scheme construction, leaving some details to be clarified, such as
\begin{itemize}
  \item[(1)] How many desired symbols are queried directly from each server in the combing part?
  \item[(2)] What are the parameters for the MDS codes used in the mixing-expanding part?
\end{itemize}
In the next sections we will fill in the details.

\subsection{Formal description of the general scheme}\label{sec3b}
Our scheme can be seemed as an extension of the scheme in \cite{Sun&Jafar16:OptimalPIR} from the case of $T=1$ to the case of general $T$. As pointed out in \cite{Sun&Jafar16:OptimalPIR}, a key idea for reducing sub-parcketization is to eliminate the symmetry across servers that was required in \cite{Sun&Jafar16:CapacityPIR,Sun&Jafar16:ColludPIR}. For example, in Fig. \ref{fg1} one can see that two symbols are queried from ${\rm Serv}^{(1)}$ and ${\rm Serv}^{(2)}$ each while only one symbol is queried from ${\rm Serv}^{(3)}$.
But for the requirement of symmetry across servers, it is supposed to query the same number of symbols from each server. However, designing for each server a special query sequence makes the scheme too complicated. It turns out dividing the servers into two disjoint groups and enforcing symmetry across the servers within each group suffice to reduce sub-parcketization and to achieve capacity simultaneously. In our scheme, the first $T$ servers are in one group and the remaining $N-T$ servers are in the other. Actually, one can divide the servers arbitrarily as long as one group contains $T$ servers and the other contains the rest servers.

We now start to give a general description of the scheme. As in the examples, we describe the scheme in two parts.

{\it Part $1$: Mixing-Expanding.} Denote the $M$ records  as $W_i\in\mathbb{F}_q^{L}$, $1\leq i\leq M$. Suppose the user wants $W_\theta$ for some $\theta\in[M]$. Then define
{\small\begin{eqnarray}
&& U_\theta=(u_{\theta,1},...,u_{\theta,L})=W_\theta S_\theta\label{eqinvert}\\
&& U_i=(u_{i,1},...,u_{i,L})=W_iS_i[:,(1:\ell)]G,~\forall i\in[M]-\{\theta\},\nonumber
\end{eqnarray}}
where $S_1,...,S_M$ are the random matrices chosen by the user uniformly from all $L\times L$ invertible matrices over $\mathbb{F}_q$. Here we leave $\ell$ and the specific form of $G$ to be determined in Section \ref{secmixpart}. One needs to know that, for the desired record $W_\theta$, the symbols $u_{\theta,1},...,u_{\theta,L}$ are independent, while for each undesired record $W_i$, $i\in[M]-\{\theta\}$, the symbols $u_{i,1},...,u_{i,L}$ are expanded from $\ell$ independent symbols through  MDS encoding.

{\it Part $2$: Combining.}
For $1\leq k\leq M$ and for a subset $\Lambda\subseteq[M]$ with $|\Lambda|=k$, we call the sum $\sum_{i\in\Lambda}u_{i,j_i}$ a {\it $\Lambda$-type $k$-sum}. The queries to each server are composed of these $k$-sums. Moreover, because of applying Step (a1), the queries to each server contain the same number of $k$-sums for every type $\Lambda$ with $|\Lambda|=k$, and we denote this number as $\gamma_k^{(i)}$ for server ${\rm Serv}^{(i)}$. For example, in Example \ref{eg2} we have $\gamma_1^{(i)}=\gamma_2^{(i)}=1, \gamma_3^{(i)}=0$ for $i=1,2$ and $\gamma_1^{(3)}=2, \gamma_2^{(3)}=0, \gamma_3^{(3)}=1$. Recall that in our scheme, the servers are divided into two groups and the servers within each groups are treated symmetrically. Therefore, for $1\leq k\leq M$,
\begin{align*}
  \gamma_k^{(1)}=\gamma_k^{(2)}=\cdots=\gamma_k^{(T)}\triangleq\alpha_k\;,\\
  \gamma_k^{(T+1)}=\gamma_k^{(T+2)}=\cdots=\gamma_k^{(N)}\triangleq\beta_k\;.
\end{align*}

In order to provide a uniform arrangement of the $k$-sums in queries, we introduce a label $(k,\delta,\rho)$ for each type $\Lambda\subseteq[M]$, where $k=|\Lambda|$, and $$\left\{\begin{array}{ll}\delta=1, ~1\leq \rho \leq\binom{M-1}{k-1}&\mbox{if $\theta\in \Lambda$}\\\delta=0, ~1\leq \rho \leq\binom{M-1}{k}&\mbox{if $\theta\not\in \Lambda$}\end{array}\right.\;.$$ Actually, $\rho$ is a counting index for all subsets $\Lambda\subseteq[M]$ with the given label $(k,\delta)$.

Then in Table \ref{tab1} we give a general description of the queries in our scheme, and in Table \ref{tab2} we further explain it by combing with Example \ref{eg2} in the case of $\theta=1$. Note that in Table \ref{tab1}, $q^{(i)}_{\Lambda,h}$ denotes a $\Lambda$-type sum, i.e. $q^{(i)}_{\Lambda,h}=\sum_{\nu\in\Lambda}u_{\nu,[\cdot]}$. Here $h$ runs from $1$ to $\alpha_{|\Lambda|}$ (or $\beta_{|\Lambda|}$), as a counting index for the $\Lambda$-type sums queried from ${\rm Serv}^{(i)}$. Inside the $\Lambda$-type sum $\sum_{\nu\in\Lambda}u_{\nu,[\cdot]}$, we do not specify the second subscript of the symbols while use $[\cdot]$ instead, because these symbols are drawn successively from $U_{\nu}$. That is, for each $k$-sum that involves $U_{\nu}$ we draw a symbol from $U_{\nu}$ that has not been used in previous queries. For example, in Table \ref{tab2} the subscripts in $a_i,b_i,c_i$'s are drawn successively from $1$ to $9$.

\begin{table}[ht]
\centering
\begin{tabular}{|l|l|l|l|l|}
  \hline
  \multicolumn{3}{|c|}{label of $\Lambda,~~(k,\delta,\rho)$}& ${\rm Serv}^{(i)},~1\leq i\leq T$ & ${\rm Serv}^{(j)},~T< j\leq N$\\\hline
  \multirow{6}{*}{$k\in[M]$}&\multirow{3}{*}{$\delta=1$}&$\rho=1$&{\scriptsize$q^{(i)}_{\Lambda,1},...,q^{(i)}_{\Lambda,\alpha_k}$}
  &{\scriptsize$q^{(j)}_{\Lambda,1},...,q^{(j)}_{\Lambda,\beta_k}$}
  \\\cline{3-5}
  &&$\tiny\cdots$&$\tiny\cdots$&$\tiny\cdots$\\\cline{3-5}&&$\rho=\binom{M-1}{k-1}$ &{\scriptsize$q^{(i)}_{\Lambda,1},...,q^{(i)}_{\Lambda,\alpha_k}$}
  &{\scriptsize$q^{(j)}_{\Lambda,1},...,q^{(j)}_{\Lambda,\beta_k}$}
  \\\cline{2-5}&\multirow{3}{*}{$\delta=0$}&$\rho=1$&{\scriptsize$q^{(i)}_{\Lambda,1},...,q^{(i)}_{\Lambda,\alpha_k}$}
  &{\scriptsize$q^{(j)}_{\Lambda,1},...,q^{(j)}_{\Lambda,\beta_k}$}
  \\\cline{3-5}
  &&$\tiny\cdots$&$\tiny\cdots$&$\tiny\cdots$
  \\\cline{3-5}&&$\rho=\binom{M-1}{k}$ &{\scriptsize$q^{(i)}_{\Lambda,1},...,q^{(i)}_{\Lambda,\alpha_k}$}
  &{\scriptsize$q^{(j)}_{\Lambda,1},...,q^{(j)}_{\Lambda,\beta_k}$}\\\hline
\end{tabular}
\caption{}\label{tab1}
\end{table}
\begin{table}[ht]
\centering
\begin{tabular}{|l|l|l|l|l|l|}
  \hline
  \multicolumn{3}{|c|}{label of $\Lambda,~~(k,\delta,\rho)$}& ${\rm Serv}^{(1)}$& ${\rm Serv}^{(2)}$& ${\rm Serv}^{(3)}$\\\hline
  \multirow{3}{*}{$k=1$}&$\delta=1$&$\Lambda=\{1\}$&$a_1$&$a_2$
  &$a_3,~a_4$
  \\\cline{2-6}
  &\multirow{2}{*}{$\delta=0$}&$\Lambda=\{2\}$&$b_1$&$b_2$&$b_3,~b_4$\\\cline{3-6}&&$\Lambda=\{3\}$ &$c_1$&$c_2$
  &$c_3,~c_4$
  \\\hline
  \multirow{3}{*}{$k=2$}&\multirow{2}{*}{$\delta=1$}&$\Lambda=\{1,2\}$&$a_{5}+b_{5}$&$a_{6}+b_{6}$
  &\\\cline{3-6}
  &&$\Lambda=\{1,3\}$&$a_{7}+c_{5}$&$a_{8}+c_{6}$&\\\cline{2-6}&$\delta=0$&$\Lambda=\{2,3\}$ &$b_{7}+c_{7}$&$b_{8}+c_{8}$
  &
  \\\hline$k=3$&$\delta=1$&$\Lambda=\{1,2,3\}$&&&$a_9+b_9+c_9$\\\hline
\end{tabular}
\caption{In Example \ref{eg2}, ($\alpha_1=\alpha_2=1,\alpha_3=0; \beta_1=2,\beta_2=0,\beta_3=1.$)}\label{tab2}
\end{table}

Therefore, once given the values of $\alpha_k,\beta_k, 1\leq k\leq M$,  one can immediately write down all queries in the scheme based on Table \ref{tab1}. The queries are usually arranged in the order that $k$ runs from $1$ to $M$, while for each value of $k$ one can arrange the $k$-sums according to an arbitrary order of the $\Lambda$'s because all types $\Lambda\subseteq[M]$ with $|\Lambda|=k$ will occur.

But in Table \ref{tab1} and Table \ref{tab2}  the queries for each $k$ are arranged according to whether $\delta=1$ or $\delta=0$. The reason for this arrangement is that we want to differentiate the mixed symbols (or desired symbols for $k=1$) from the undesired symbols for the convenience of latter analysis. It is important to note that
the order of $\Lambda$'s should reveal no information on $\theta$, so in actual implementation the user arranges the queries according to an arbitrarily given order of all types $\Lambda\subseteq[M]$ with $|\Lambda|=k$.

Then we determine the parameters $\alpha_k,\beta_k$ for the combining part and $\ell,~G$ for the  mixing-expanding part respectively in the next two sections.

\subsection{Parameters in the combining part}\label{sec3c}
The parameters $\alpha_k$ and $\beta_k$ are computed as a result of establishing sufficient conditions for the correctness requirement and the privacy requirement of the PIR scheme.
\begin{itemize}
  \item[(s1)] {\it A sufficient condition for the correctness requirement: }
\vspace{4pt}

  For any $\Lambda\subseteq[M]-\{\theta\}$, denote $\underline{\Lambda}=\Lambda\cup\{\theta\}$. Suppose $|\Lambda|=k$, and denote{\small$$\begin{array}{c}
  Q_{\Lambda}=(q^{(1)}_{\Lambda,[1:\alpha_k]},...,q^{(T)}_{\Lambda,[1:\alpha_k]},q^{(T+1)}_{\Lambda,[1:\beta_k]},...,q^{(N)}_{\Lambda,[1:\beta_k]})\\
  Q_{\underline{\Lambda}}=(q^{(1)}_{\underline{\Lambda},[1:\alpha_{k+1}]},...,q^{(T)}_{\underline{\Lambda},[1:\alpha_{k+1}]},q^{(T+1)}_{\underline{\Lambda},[1:\beta_{k+1}]},...,q^{(N)}_{\underline{\Lambda},[1:\beta_{k+1}]})
  \end{array}$$}Then $Q_{\Lambda}$ is a $(T\alpha_k+(N-T)\beta_k)$-dimensional vector and $Q_{\underline{\Lambda}}$ is a $(T\alpha_{k+1}+(N-T)\beta_{k+1})$-dimensional vector. The sufficient condition is that
  \begin{equation}\label{eqs1}
    (Q_{\Lambda},Q_{\underline{\Lambda}}-\widetilde{U}_{\theta})=Q_{\Lambda}G_k\;,
  \end{equation}
  where $G_k$ is a generator matrix of an MDS code with dimension $T\alpha_k+(N-T)\beta_k$ and information rate $\frac{T}{N}$, and $Q_{\underline{\Lambda}}-\widetilde{U}_{\theta}$ means eliminating the symbols of $U_\theta$ from the coordinates of $Q_{\underline{\Lambda}}$.
\end{itemize}
Eg. in Table \ref{tab2}, for $\Lambda=\{3\}$ it requires that $(c_1,c_2,...,c_6)$ is a codeword in a $[6,4]$ MDS code; for $\Lambda=\{2,3\}$ it requires that $(b_7+c_7,b_8+c_8,b_9+c_9)$ is a codeword in a $[3,2]$ MDS code.
We first explain why the condition (s1) is sufficient for the correctness requirement. Since the user can directly get $Q_{\Lambda}$ and $Q_{\underline{\Lambda}}$ from the servers, he can first recover $Q_{\underline{\Lambda}}-\widetilde{U}_{\theta}$ from (\ref{eqs1}) and then derive $\widetilde{U}_{\theta}$ which contains a part of desired symbols. In a similar way, for $\Lambda$ running through all subsets in $[M]-\{\theta\}$, the user can recover all desired symbols involved in mixed symbols, in addition to the  desired symbols  directly received  from servers when $k=1$, the user can finally collect all the desired symbols needed to recover the desired record $W_\theta$. Note that the total number of the desired symbols (including those appear in mixed symbols) equals the number of $\Lambda$-type sums with $\delta(\Lambda)=1$.  Thus the invertible transformation in (\ref{eqinvert}) implies that
{\small\begin{equation}\label{eqL}L=\sum_{k=1}^{M}\binom{M-1}{k-1}(T\alpha_k+(N-T)\beta_k)\;.\end{equation}}

Because the MDS code we used in (s1) has information rate $\frac{T}{N}$, then we can derive a relation between $\alpha_k$ and $\beta_k$ from the sufficient condition (\ref{eqs1}), i.e., for $1\leq k<M$,
{\small\begin{equation}
  (T\alpha_k+(N-T)\beta_k)\frac{N-T}{T}=T\alpha_{k+1}+(N-T)\beta_{k+1}.
\end{equation}}

\vspace{6pt}

\begin{itemize}
  \item[(s2)] {\it A sufficient condition for the privacy requirement:}

For any subset $\Gamma\subset[N]$ with $|\Gamma|\leq T$, denote by $U_i^{\Gamma}$, $1\leq i\leq M$, the vector consisting of symbols from $U_i$ that are involved in the queries to servers in $\Gamma$. The sufficient condition is that $U_1^{\Gamma},...,U_M^{\Gamma}$ are independent vectors consisting of the same number of coordinates drawn independently and uniformly from $\mathbb{F}_q$.
\end{itemize}
Eg. in Example \ref{eg2} suppose $\Gamma=\{2,3\}$, then we have $U_1^{\Gamma}=(a_2,a_3,a_4,a_6,a_8,a_9)$, $U_2^{\Gamma}=\{b_2,b_3,b_4,b_6,b_8,b_9\}$ and $U_3^{\Gamma}=\{c_2,c_3,c_4,c_6,c_8,c_9\}$. According to (s2) it requires that $U_1^{\Gamma}, U_2^{\Gamma}, U_3^{\Gamma}$ are independent and each one is composed of $6$ independent symbols.
Recall in Remark \ref{re2} we assume that for each record $W_i=(w_{i,1},...,w_{i,L})\in\mathbb{F}_q^L$, the symbols $w_{i,1},...,w_{i,L}$ are drawn independently and uniformly from $\mathbb{F}_q$. Moreover, $W_i$ and $W_j$ are independent for $i\neq j$. The following proposition states an evident fact and we omit its proof.
\begin{proposition}\label{propind}
Suppose $W_i=(w_{i,1},...,w_{i,L})\in\mathbb{F}_q^L$ is composed of $L$ symbols drawn independently and uniformly from $\mathbb{F}_q$, and $(u_{i,1},...,u_{i,h})=W_iS$ for some matrix $S\in\mathbb{F}_q^{L\times h}$, $1\leq h\leq L$. Then $u_{i,1},...,u_{i,h}$ are independent if and only if $S$ has rank $h$.
\end{proposition}
    We now explain why (s2) is sufficient for the privacy requirement. For any colluding subset $\Gamma\subset[N]$ with $|\Gamma|\leq T$, the queries to servers in $\Gamma$ are computed as $(U_1^{\Gamma},U_2^{\Gamma},...,U_M^{\Gamma})C$ after the combining part, where $C$ is the combining matrix which is fixed for given values of $M,N,T$. Moreover, $U_i^{\Gamma}=W_i\tilde{S}^{\Gamma}_i$ for some matrix $\tilde{S}^{\Gamma}_i$, then to clarify the privacy requirement it is sufficient to show that $\tilde{S}^{\Gamma}_1,...,\tilde{S}^{\Gamma}_M$ are identically distributed. By the condition (s2) and Proposition \ref{propind}, it implies that $\tilde{S}^{\Gamma}_1,...,\tilde{S}^{\Gamma}_M$ are all composed of the same number of linearly independent columns. Since $\tilde{S}^{\Gamma}_i$'s are actually generated from the random matrices $S_i$'s which are unknown to the servers, the colluding subset $\Gamma$ cannot tell the difference between the $\tilde{S}^{\Gamma}_i$'s. Therefore, the privacy requirement is satisfied.

According to (\ref{eqs1}) the undesired symbols are expanded by using independent MDS codes, so dependence may only happen within the same MDS codes. However, from the MDS property and Proposition \ref{propind}, any $T$ symbols of a codeword of an $[N,T]$ MDS code are independent. So based on the general scheme, a sufficient condition for (s2) is that for $1\leq k<M$, and for any $x\in\mathbb{N}$ with $0\leq x\leq T$,
\begin{equation}\label{eqpriv}{\small
  x(\alpha_k+\alpha_{k+1})+(T-x)(\beta_k+\beta_{k+1})\leq T\alpha_k+(N-T)\beta_k}.
\end{equation}
Actually in (\ref{eqpriv}), $x\triangleq|\Gamma\cap[T]|$, thus the left side of (\ref{eqpriv}) is the number of symbols from each undesired record involved in the codeword $(Q_{\Lambda},Q_{\underline{\Lambda}}-\widetilde{U}_{\theta})$ which according to (\ref{eqs1}) is  an MDS codeword, and the right side of (\ref{eqpriv}) equals the dimension of the MDS code.

Therefore, $\alpha_k,\beta_k$ are actually the integral solutions to (\ref{eqs1}) and (\ref{eqpriv}). We further find out the condition
$$\left\{\begin{array}{l}T\alpha_{k+1}=(N-T)\beta_k\\\alpha_k+\alpha_{k+1}=\beta_k+\beta_{k+1}\end{array}\right.$$
implies (\ref{eqs1}) and (\ref{eqpriv}) simultaneously. Thus we only need to solve the following simplified equations on $\alpha_k$ and $\beta_k$, i.e.
\begin{equation}\label{eq}
\left\{\begin{array}{l}T\alpha_{k+1}=(N-T)\beta_k\\\alpha_k+\alpha_{k+1}=\beta_k+\beta_{k+1}\\\alpha_k,\beta_k\in \mathbb{N},~~1\leq k< M\end{array}\right.
\end{equation}

From the recursive relations in (\ref{eq}) we get two geometric series, i.e.,
\begin{equation}\label{eqser}{\small
\left\{\begin{array}{l}
T\alpha_{k+1}+(N-T)\beta_{k+1}=(\frac{N-T}{T})^k(T\alpha_1+(N-T)\beta_1)\\
\alpha_{k+1}-\beta_{k+1}=(-1)^k(\alpha_1-\beta_1)
\end{array}\right.}
\end{equation}
By giving proper initial values to $\alpha_1$ and $\beta_1$, we can finally get the solutions to (\ref{eq}). Specifically, when $N\geq 2T$, the first series in (\ref{eqser}), i.e., $\{T\alpha_{k}+(N-T)\beta_{k}\}_{1\leq k\leq M}$ is increasing, so we set $\alpha_1=t^{M-2}$ and $\beta_1=0$ to ensure $\alpha_k,\beta_k\in \mathbb{N}$ for $1\leq k< M$. Under this initial values, we get the solutions
\begin{equation}\label{eqsolution1}
\left\{\begin{array}{l}\alpha_k= \frac{(n-t)^{k-2}-(-t)^{k-2}}{n}(n-t)t^{M-k} \\
    \beta_k = \frac{(n-t)^{k-1}-(-t)^{k-1}}{n}t^{M-k}
\end{array}\right.
\end{equation}
where $n=\frac{N}{d}, t=\frac{T}{d}$ and $d={\rm gcd}(N,T)$. When $T<N<2T$, the series $\{T\alpha_{k}+(N-T)\beta_{k}\}_{1\leq k\leq M}$ is decreasing, so we alternatively set $\alpha_M=0,\beta_M=(n-t)^{M-2}$, and get the solutions
\begin{equation}\label{eqsolution2}
\left\{\begin{array}{l}\alpha_k= \frac{t^{M-k}-(t-n)^{M-k}}{n}(n-t)^{k-1} \\
    \beta_k = \frac{t^{M-k-1}-(t-n)^{M-k-1}}{n}t(n-t)^{k-1}
\end{array}\right.
\end{equation}
In summary, (\ref{eqsolution1}) and (\ref{eqsolution2}) respectively give the solution to (\ref{eq}) for the case $N\geq 2T$ and  the case $T<N<2T$. In both cases, one can check that
\begin{equation}\label{eqeq}
T\alpha_{k}+(N-T)\beta_{k}=d(n-t)^{k-1}t^{M-k}\;.
\end{equation}

\subsection{Parameters in the mixing-expanding part }\label{secmixpart}
In this section we determine the parameter $\ell$ and the matrix $G$ in (\ref{eqinvert}). This is equivalent to determining the MDS codes used in the mixing-expanding part for the undesired records.

For any undesired index $i\in[M]-\{\theta\}$, denote $\mathcal{I}_i=\{\Lambda\mid\Lambda\subseteq[M]-\{\theta\},~i\in\Lambda\}$. Thus the set
\begin{equation}\label{eqMDS}\{(Q_{\Lambda},Q_{\underline{\Lambda}})\mid\Lambda\in\mathcal{I}_i\}\end{equation}covers all coordinates of $U_i$ and each coordinate appears exactly once. This implies that
\begin{equation}\label{eqL2}{\small
L=\sum_{k=1}^{M-1}\binom{M-2}{k-1}(T(\alpha_k+\alpha_{k+1})+(N-T)(\beta_k+\beta_{k+1}))
}\end{equation}
Since our solutions of $\alpha_k,\beta_k$ satisfying (\ref{eqeq}), one can check that (\ref{eqL2}) coincides with (\ref{eqL}) , while the former computes $L$ by counting the number of undesired symbols and the latter computed by counting the number of desired symbols. This consistence just verifies application of Step (a1) throughout the scheme implementation.

Moreover, as required by the condition (s1), $(Q_{\Lambda},Q_{\underline{\Lambda}}-\widetilde{U}_{\theta})$ is a codeword of an MDS code. Thus in (\ref{eqMDS}) a total number of $\sum_{k=1}^{M-1}\binom{M-2}{k-1}=2^{M-2}$ MDS codes are involved. Actually, these MDS codes can be indexed by $\mathcal{I}_i$. From (\ref{eqs1}) we also know that for each $\Lambda\in\mathcal{I}_i$ with $|\Lambda|=k$, the corresponding MDS code has dimension $T\alpha_k+(N-T)\beta_k$ and information rate $\frac{T}{N}$, and we denote its generator matrix as $G_k$. Furthermore, according the condition (s2), these MDS codewords must be mutually independent.

However, for any $\Lambda\in\mathcal{I}_i$, each coordinate of $(Q_{\Lambda},Q_{\underline{\Lambda}}-\widetilde{U}_{\theta})$ is a $\Lambda$-type sum. By linearity of the MDS code we used, it is sufficient to require that for all $j\in\Lambda$, the corresponding coordinates of $U_j$ form a codeword from the same MDS code. For example, in table \ref{tab2} by the condition (s1) it requires that $(b_7+c_7,b_8+c_8,b_9+c_9)$ is an MDS codewords, for which it is sufficient to require that both $(b_7,b_8,b_9)$ and $(c_7,c_8,c_9)$ are codewords from the same MDS code.

Therefore, if we arrange the elements in $\mathcal{I}_i$ according to an increasing order of cardinality, i.e., $\mathcal{I}_i=\{\Lambda_1,...,\Lambda_{2^{M-2}}\}$ where $|\Lambda_{j}|\leq|\Lambda_{j'}|$ for $j<j'$, and rearrange the coordinates of $U_i$ correspondingly, i.e.,
\begin{equation}\label{equ}{\small U_i=(U_{i,\Lambda_1},U_{i,\underline{\Lambda}_1},U_{i,\Lambda_2},U_{i,\underline{\Lambda}_2},...,U_{i,\Lambda_{2^{M-2}}},U_{i,\underline{\Lambda}_{2^{M-2}}})}\end{equation}
where $U_{i,\Lambda_j}$ contains the coordinates of $U_i$ that appear in  $Q_{\Lambda_j}$ and $\underline{\Lambda}_j=\Lambda_j\cup\{\theta\}$ for $1\leq j\leq 2^{M-2}$, then one can see
\begin{equation}\label{eqg1}U_i=(U_{i,\Lambda_1},U_{i,\Lambda_2},...,U_{i,\Lambda_{2^{M-2}}})\widetilde{G}\end{equation} where
\begin{equation}\label{eqg2}
\widetilde{G}={\rm diag}[I_1\otimes G_1,I_{\binom{M-2}{1}}\otimes G_2,...,I_{\binom{M-2}{k-1}}\otimes G_k,...,I_1\otimes G_{M-1}].
\end{equation}
Here, ${\rm diag}[A,B]$ denotes a diagonal block matrix with $A$ and $B$ lying on the diagonal, $I_j$ denotes the $j\times j$ identity matrix, and recall that $G_k$ is a generator matrix of a fixed MDS code with dimension $T\alpha_k+(N-T)\beta_k$ and information rate $\frac{T}{N}$. Because for each $k\in[M-1]$, there are $\binom{M-2}{k-1}$ $\Lambda$'s in $\mathcal{I}_i$ with $|\Lambda|=k$, and each $\Lambda$ corresponds to an MDS code with generator matrix $G_k$, thus $I_{\binom{M-2}{k-1}}\otimes G_k$ denotes the MDS encoding of the corresponding coordinates in the $\Lambda$-type sums with $|\Lambda|=k$. From (\ref{eqg1}) and (\ref{eqg2}), one can see that $U_i$ is actually the direct sum of all the independent MDS codewords, which coincides with the requirement of  (s1) and (s2). Finally, the order that we take symbols from $U_i$ in the combining part as described in Section \ref{sec3b} is a little different from that presented in (\ref{equ}), so the actual form of $G$ can be obtained from $\widetilde{G}$ in (\ref{eqg2}) through a proper column permutation.

From (\ref{eqg1}), one can compute the value of $\ell$, i.e.
{\small\begin{eqnarray}
 \ell &=& \sum_{j=1}^{2^{M-2}}|\Lambda_j|\nonumber\\
  &=&\sum_{k=1}^{M-1}\binom{M-2}{k-1}(T\alpha_k+(N-T)\beta_k)\nonumber \\
  &\stackrel{(a)}{=}&d\sum_{k=1}^{M-1}\binom{M-2}{k-1}(n-t)^{k-1}t^{M-k}\nonumber\\
  &=&Tn^{M-2}\nonumber
\end{eqnarray}}
where {\small$(a)$} is from (\ref{eqeq}).

\begin{remark}
Note that $\mathbb{F}_q$ is the base field for this scheme, and its size depends on  parameters of the MDS codes used in the mixing-expanding part. From (\ref{eqg2}) one can see a total of $M-1$ MDS codes over $\mathbb{F}_q$ have been used, each of which has length $\frac{N}{T}(T\alpha_k+(N-T)\beta_k)$, $1\leq k<M$. Therefore, combining with (\ref{eqeq}) it must have $$q\geq \frac{N}{T}(T\alpha_k+(N-T)\beta_k)=N(n-t)^{k-1}t^{M-k-1}.$$
It is sufficient to require $q\geq\max_{1\leq k<M}\{N(n-t)^{k-1}t^{M-k-1}\}=\max\{Nt^{M-2},N(n-t)^{M-2}\}$. One can see the capacity-achieving scheme in \cite{Sun&Jafar16:ColludPIR} works over a finite field $\mathbb{F}_q$ with $q\geq\max\{N^2T^{M-2},N^2(N-T)^{M-2}\}$. Thus we reduce the field size by a factor of $\frac{1}{Nd^{M-2}}$.

\end{remark}

\subsection{Properties of the scheme}
In this section, we verify the scheme described in Section \ref{sec3b} has properties: (1) correctness and privacy as required in Section \ref{sec2}; (2) optimal sub-packetization; (3) achieving the capacity.

Because we have stated in Section \ref{sec3c} that the conditions (s1) and (s2) are sufficient for the requirement of correctness and privacy, and the parameters we determined in Section \ref{sec3c} and Section \ref{secmixpart} ensure that our scheme can satisfy (s1) and (s2), thus our PIR scheme satisfies correctness and privacy.

Then we compute subpacektization of the scheme, that is, computing the value of $L$. By (\ref{eqL}) and (\ref{eqeq}), we have
\begin{equation*}
L=\sum_{k=1}^M\binom{M-1}{k-1}d(n-t)^{k-1}t^{M-k}
=dn^{M-1}\;,
\end{equation*}
attaining the lower bound on $L$ we proved in Theorem \ref{thm2}. As a result, our scheme achieves the optimal sub-packetization.

Finally, the rate of our scheme equals $\frac{L}{D}$, where $D$ is number of symbols downloaded from all servers. From Table \ref{tab1} one can see
\begin{eqnarray*}
D&=&\sum_{k=1}^M\binom{M}{k}(T\alpha_k+(N-T)\beta_k)\\
&=&\sum_{k=1}^M\binom{M}{k}d(n-t)^{k-1}t^{M-k}\\
&=&\frac{d(n^M-t^M)}{n-t}\;.
\end{eqnarray*}
Therefore, the rate equals $\frac{L}{D}=dn^{M-1}\times\frac{n-t}{d(n^M-t^M)}=\frac{1-\frac{T}{N}}{1-(\frac{T}{N})^M}$ attaining capacity for the PIR schemes.

\begin{corollary}
For the nontrivial case, i.e., $M\geq 2$ and $1\leq T<N$, the optimal sub-packetization for a linear capacity-achieving $T$-private PIR scheme for $N$ replicated servers and $M$ records, is $dn^{M-1}$, where $d={\rm gcd}(N,T)$ and $n=\frac{N}{d}$.
\end{corollary}

\section{Conclusions}
Sub-packetization is an important parameter related to the implementation complexity of a scheme.
In this work we determine the optimal sub-packetization for linear capacity-achieving schemes with replicated servers. The relations between sub-packetization, download rate, and field size for PIR schemes deserve further studies in the future.


\appendices
\section{Proofs for Lemma \ref{lem2} in Section \ref{sec2b}}\label{appenA}
\noindent{\bf Lemma~\ref{lem2}.}
For any $\theta,\theta^\prime\in[M]$, and for any subset $\Lambda\subseteq[M]$ and $\Gamma\subseteq[N]$ with $|\Gamma|\leq T$,
\begin{equation*}
H(\mathcal{A}^{\Gamma}_{\theta}| \mathcal{W}_\Lambda,\mathcal{Q})=H(\mathcal{A}^{\Gamma}_{\theta^\prime}| \mathcal{W}_\Lambda,\mathcal{Q})\;.\end{equation*}
\begin{proof}
From the privacy condition (\ref{eq3}), we have that $I(\theta;\mathcal{W}_\Lambda,\mathcal{Q}^{\Gamma}_\theta)=0$, implying
$H(\mathcal{W}_\Lambda,\mathcal{Q}^{\Gamma}_\theta)=H(\mathcal{W}_\Lambda,\mathcal{Q}^{\Gamma}_{\theta^\prime})$.
Similarly, $H(\mathcal{A}^{\Gamma}_\theta,\mathcal{W}_\Lambda,\mathcal{Q}^{\Gamma}_\theta)=H(\mathcal{A}^{\Gamma}_{\theta^\prime},\mathcal{W}_\Lambda,\mathcal{Q}^{\Gamma}_{\theta^\prime})$.
It immediately follows that $H(\mathcal{A}^{\Gamma}_{\theta}| \mathcal{W}_\Lambda,\mathcal{Q}^{\Gamma}_\theta)=H(\mathcal{A}^{\Gamma}_{\theta^\prime}| \mathcal{W}_\Lambda,\mathcal{Q}^{\Gamma}_{\theta^\prime})$.

Then the lemma follows from the claim below:

{\it Claim: }For any $\theta\in[M]$, and for any subset $\Lambda\subseteq[M]$ and $\Gamma\subseteq[N]$,
$H(\mathcal{A}^{\Gamma}_\theta| \mathcal{W}_\Lambda,\mathcal{Q}^{\Gamma}_\theta)=H(\mathcal{A}^{\Gamma}_\theta| \mathcal{W}_\Lambda,\mathcal{Q})$.

{\it Proof of the claim.}
On the one hand, it obviously holds $H(\mathcal{A}^{\Gamma}_\theta| \mathcal{W}_\Lambda,\mathcal{Q}^{\Gamma}_\theta)\geq H(\mathcal{A}^{\Gamma}_\theta| \mathcal{W}_\Lambda,\mathcal{Q})$. On the other hand ,
{\small\begin{eqnarray*}
&&H(\mathcal{A}^{\Gamma}_\theta| \mathcal{W}_\Lambda,\mathcal{Q}^{\Gamma}_\theta)-H(\mathcal{A}^{\Gamma}_\theta| \mathcal{W}_\Lambda,\mathcal{Q})\\
&=&I(\mathcal{A}^{\Gamma}_\theta;\mathcal{Q}\setminus \mathcal{Q}^{\Gamma}_\theta|\mathcal{W}_\Lambda,\mathcal{Q}^{\Gamma}_\theta)\\
&\leq& I(\mathcal{A}^{\Gamma}_\theta,\mathcal{W}_{[M]-\Lambda} ;\mathcal{Q}\setminus \mathcal{Q}^{\Gamma}_\theta|\mathcal{W}_\Lambda,\mathcal{Q}^{\Gamma}_\theta)\\
&=&I(\mathcal{W}_{[M]-\Lambda} ;\mathcal{Q}\setminus \mathcal{Q}^{\Gamma}_\theta|\mathcal{W}_\Lambda,\mathcal{Q}^{\Gamma}_\theta)+ I(\mathcal{A}^{\Gamma}_\theta;\mathcal{Q}\setminus \mathcal{Q}^{\Gamma}_\theta|\mathcal{W}_{[M]},\mathcal{Q}^{\Gamma}_\theta)\\
&\stackrel{(a)}{=}&0\;.
\end{eqnarray*}}
Note the equality in {\small$(a)$} comes from the facts $I(\mathcal{W}_{[M]-\Lambda} ;\mathcal{Q}\setminus \mathcal{Q}^{\Gamma}_\theta|\mathcal{W}_\Lambda,\mathcal{Q}^{\Gamma}_\theta)=0$ and $I(\mathcal{A}^{\Gamma}_\theta;\mathcal{Q}\setminus \mathcal{Q}^{\Gamma}_\theta|\mathcal{W}_{[M]},\mathcal{Q}^{\Gamma}_\theta)=0$, where the former is due to  (\ref{eq0}) and the latter is from
(\ref{eq1}).
\end{proof}

\section{Proof of the Lemmas in Section \ref{secc}}\label{appenB}
\noindent{\bf Lemma~\ref{lem5}.}
Let $a,b,m\in\mathbb{N}$. Suppose $d_1={\rm gcd}(a,b),d_2={\rm gcd}(a^m,\sum^{m}_{i=0}a^{m-i}b^i)$, then $d_2=d_1^m$.
\vspace{4pt}

\begin{proof}
Obviously $d_1^m|d_2$. Let $a=a_1d_1,b=b_1d_1,c=\frac{d_2}{d_1^m}$, then ${\rm gcd}(a_1,b_1)=1$ and $c|a_1^m,c|\sum^{m}_{i=0}a^{m-i}_1b^i_1$. If $c\neq 1$, then for any  prime factor $p$ of $c$, it must have $p|a_1^m$ and $p|\sum^{m}_{i=0}a^{m-i}_1b^i_1$, implying $p|a_1$ and $p|b_1$, which contradicts the fact ${\rm gcd}(a_1,b_1)=1$.
\end{proof}

\vspace{6pt}

\noindent{\bf Lemma~\ref{thm3}.}
Let $V_1,...,V_N$ be $N$ linear subspaces of $\mathbb{F}^{L}$ and denote $V=\sum^N_{i=1}V_i$. Suppose $T\in\mathbb{N}$ and $1\leq T<N$. If for any subset $\Gamma\subseteq[N]$ with $|\Gamma|=T$ it holds $\dim(\sum_{i\in\Gamma}V_i)=\frac{T}{N}\dim V$, then
$N|\dim V$ and $\dim(V_i)=\frac{1}{N}\dim V$ for $1\leq i\leq N$. Moreover, $V=\oplus^N_{i=1}V_i$.

\vspace{4pt}

\begin{proof}
We first give two claims which hold under the hypothesis of the lemma:

{\it Claim~1:  There exists $i_0\in[N]$ such that ${\rm dim}V_{i_0}\geq\lceil\frac{{\rm dim}V}{N}\rceil$.}

{\it Claim~2: For $1\leq s\leq T$, for any subset $\Gamma'\subseteq[N]$ with $|\Gamma'|=s$, it holds
\begin{equation}\label{eqlinear}{\rm dim}\sum_{i\in\Gamma'}V_i\leq\frac{T}{N}{\rm dim}V-(T-s)\lceil\frac{{\rm dim}V}{N}\rceil\;.\end{equation}}

Note that Claim 1 is a direct result of the fact $\dim V\leq\sum_{i=1}^N\dim V_i$. Next we prove Claim 2 by induction on $s$. For $s=T$, (\ref{eqlinear}) holds from the hypothesis of the lemma. Then we assume that (\ref{eqlinear}) holds for some $s>1$ and prove it also holds for $s-1$.

On the contrary, we assume there exists
$\Gamma'\subseteq [N]$ with $|\Gamma'|=s-1$, such that \begin{equation}\label{eqassume}{\rm dim}\sum_{i\in\Gamma'}V_i>\frac{T}{N}{\rm dim}V-(T-s+1)\lceil\frac{{\rm dim}V}{N}\rceil\;.\end{equation} Select $\Gamma''\subseteq[N]$ such that $\Gamma'\subset\Gamma''$ and $|\Gamma''|=s$. Denote $\overline{\Gamma''}=[N]-\Gamma''$ and $V''=\sum_{i\in\Gamma''}V_i$. Since (\ref{eqlinear}) holds for $s$, it has ${\rm dim}V''\leq\frac{T}{N}{\rm dim}V-(T-s)\lceil\frac{{\rm dim}V}{N}\rceil$.
Therefore,
{\footnotesize\begin{eqnarray*}\dim\sum_{j\in\overline{\Gamma''}}V_j/(V_j\cap V'')&\geq&\dim V-{\rm dim}V''\\
&\geq&\dim V-\frac{T}{N}{\rm dim}V+(T-s)\lceil\frac{{\rm dim}V}{N}\rceil\\
&\geq&\frac{N-s}{N}\dim V.\end{eqnarray*}}
Because $|\overline{\Gamma''}|=N-s$, by a deduction similar to that of Claim 1 we have $\dim V_{j_0}/(V_{j_0}\cap V'')\geq \lceil\frac{{\rm dim}V}{N}\rceil$ for some $j_0\in\overline{\Gamma''}$. Therefore,
{\footnotesize\begin{eqnarray*}
\dim \sum_{i\in\Gamma'\cup\{j_0\}}V_i&=&\dim \sum_{i\in\Gamma'}V_i+\dim V_{j_0}/(V_{j_0}\cap\sum_{i\in\Gamma'}V_i ) \\&\stackrel{(a)}{\geq}&\dim \sum_{i\in\Gamma'}V_i+\dim V_{j_0}/(V_{j_0}\cap V'' )\\
&\stackrel{(b)}{>}&\frac{T}{N}{\rm dim}V-(T-s+1)\lceil\frac{{\rm dim}V}{N}\rceil+\lceil\frac{{\rm dim}V}{N}\rceil\\&=&\frac{T}{N}{\rm dim}V-(T-s)\lceil\frac{{\rm dim}V}{N}\rceil
\end{eqnarray*}}
where {\small $(a)$} holds because $\sum_{i\in\Gamma'}V_i\subset V''$, and {\small $(b)$} is based on the assumption (\ref{eqassume}). However, $|\Gamma'\cup\{j_0\}|=s$, thus we get a contradiction to the induction base that (\ref{eqlinear}) holds for $s$. Hence the assumption (\ref{eqassume}) is not true and Claim 2 holds correctly.

\vspace{4pt}

Set $s=1$, from Claim 2 we have for all $i\in[N]$,
 {\footnotesize$$ \dim V_i\leq \frac{T}{N}\dim V-(T-1)\lceil\frac{\dim V}{N}\rceil\leq\frac{1}{N}\dim V\leq\frac{1}{N}\sum_{i=1}^N\dim V_i,$$}
implying that $\dim V_i=\frac{1}{N}\sum_{i=1}^N\dim V_i$ for all $i\in[N]$. The lemma follows immediately.

\end{proof}

\end{document}